\documentclass[a4paper, 12pt]{article}

\usepackage[sort&compress]{natbib}
\bibpunct{(}{)}{;}{a}{}{,} 

\usepackage{amsthm, amsmath, amssymb, mathrsfs, multirow, url,bm}
\usepackage{graphicx} 
\usepackage{ifthen} 
\usepackage{amsfonts}
\usepackage[usenames]{color}
\usepackage{fullpage}
\usepackage{caption}
\usepackage{subcaption}

\theoremstyle{plain} 
\newtheorem{thm}{Theorem}

\newtheorem{lem}{Lemma}

\theoremstyle{definition}

\newtheorem{asmp}{Assumption}

\theoremstyle{remark}

\newcommand{\RR}{\mathbb{R}}

\newcommand{\E}{\mathsf{E}}

\newcommand{\prob}{\mathsf{P}}

\newcommand{\eps}{\varepsilon}

\newcommand{\tr}{\mathrm{tr}}

\newcommand{\iid}{\overset{\text{\tiny iid}}{\sim}}

\newcommand{\nm}{\mathsf{N}}
\newcommand{\N}{\mathsf{N}}

\newcommand{\chisq}{\mathsf{ChiSq}}

\title{Bayesian inference in high-dimensional linear models using an empirical correlation-adaptive prior}
\author{C.~Liu,\footnote{Department of Statistics, North Carolina State University; {\tt cliu22@ncsu.edu}, {\tt yyang44@ncsu.edu}, {\tt rgmarti3@ncsu.edu}} \; Y.~Yang,$^*$ \; H.~Bondell,\footnote{{School of Mathematics and Statistics}, University of Melbourne; {\tt howard.bondell@unimelb.edu.au}} \; and R.~Martin$^*$} 
\date{\today}

\begin{document}

\maketitle 


\begin{abstract}
In the context of a high-dimensional linear regression model, we propose the use of an empirical 
correlation-adaptive prior that makes use of information in the observed predictor variable matrix to adaptively address high collinearity, determining if parameters associated with correlated predictors should be shrunk together or kept apart. Under suitable conditions, we prove that this empirical Bayes posterior concentrates around the true sparse parameter at the optimal rate asymptotically.  A simplified version of a shotgun stochastic search algorithm is employed to implement the variable selection procedure, and we show, via simulation experiments across different settings and a real-data application, the favorable performance of the proposed method compared to existing methods.  

\smallskip

\end{abstract}

\section{Introduction}
\label{s:intro}

Consider the standard linear regression model 
\[ Y = X\beta + \eps, \]
where $Y$ is a $n \times 1$ vector of response variables, $X$ is a $n \times p$ matrix of predictor variables, $\beta$ is a $p \times 1$ vector of regression coefficients, and $\eps$ is a vector of iid $\nm(0,\sigma^2)$  errors.  We are interested in the high-dimensional case, where $p\gg n$. Furthermore, it is assumed that the  true $\beta$ is sparse, i.e., only a small subset of the $\beta$ coefficients are nonzero. 


There are a variety of methods available for estimating $\beta$ under a sparsity constraint.  This includes regularization-based methods like the lasso \citep{c2}, the adaptive lasso \citep{c3}, the SCAD \citep{c4}, and MCP \citep{c5}; see \citet{c1} for a review.  From a Bayesian point of view, varieties of priors for regression coefficients and the model space have been developed, leading to promising selection properties. For the regression coefficients, $\beta$, the normal mixture prior is specified in \citet{c6}; \citet{c9} introduce empirical Bayes ideas; \citet{c7} use spike-and-slab priors; \citet{c10} estimate $\beta$ as the ``most sparse'' among those in a suitable posterior credible region; \citet{polson.scott.2012} consider a horseshoe prior; \citet{c11} use shrinking and diffusing priors; and \cite{c12} consider an empirical Bayes version of spike-and-slab. 


Collinearity is unavoidable in high-dimensional settings, and 
methods such as lasso tend to smooth away the regression coefficients of highly collinear predictors, and hence deter correlated covariates to be included in the model simultaneously.  This motivated \citet{c15} to propose an adaptive powered correlation prior that lets the data itself weigh in on how collinear predictors are to be handled.  However, their suggested generalized Zellner's prior is not applicable in the $p > n$ scenario.  To overcome this, we use an empirical Bayes approach, and propose an {\em empirical correlation-adaptive prior} (ECAP) to this framework.  In Section~\ref{s:model}, we present our empirical Bayes model and a motivating example illustrating the effect of correlation-adaptation in the prior.  Asymptotic posterior concentration properties are derived in Section~\ref{s:inf}, in particular, adaptive minimax rates are established for estimation of the mean response.  In Section~\ref{s:comp}, we recommend a shotgun stochastic search approach for computation of the posterior distribution over the model space.  Simulation experiments are presented in Section~\ref{s:sim}, and we demonstrate the benefits of adaptively varying the correlation structure in the prior for variable selection compared to existing methods.  The real-data illustration in Section~\ref{s:real} highlights the improved predictive performance that can be achieved using the proposed correlation-adaptive prior.  Proofs of the theorems are deferred to the Appendix.

\section{Model specification}
\label{s:model}

\subsection{The prior}
\label{SS:prior}

Under assumed sparsity, it is natural to decompose $\beta$ as $(S, \beta_S)$, where $S \subseteq \{1,2,\ldots,p\}$ is the set of non-zero coefficients, called the {\em configuration} of $\beta$, and $\beta_S$ is the $|S|$-vector of non-zero values, with $|S|$ denoting the cardinality of $S$.  We will write $X_S$ for the sub-matrix of $X$ corresponding to the configuration $S$.  With this decomposition of $\beta$, a hierarchical prior is convenient, i.e., a prior for $S$ and a conditional prior for $\beta_S$, given $S$.  

First, for the prior $\pi(S)$ for $S$, we follow \citet{c12} and write 
\[ \pi(S)=\pi(S \mid |S|=s)f_n(s), \]
where $f_n(s)$ is a prior on $|S|$ and $\pi(S \mid |S|=s)$ is a conditional prior on $S$, given $|S|$.  Based on the recommendation in \citet{c13}, we take 
\begin{align}
\label{eq:fns}
    f_n(s) \propto c^{-s}p^{-as}, s=0,1,\ldots,R,
\end{align}  
where $a$ and $c$ are positive constants, and $R=\text{rank}(X) \leq n$.  It is common to take $\pi(S \mid |S|=s)$ to be uniform, but here we break from this trend to take collinearity into account.  Let $D(S) = |X_S^\top X_S|$ denote the determinant of $X_S^\top X_S$, and consider the geometric mean of the eigenvalues, $D(S)^{1/|S|}$, as a measure of the ``degree of collinearity'' in model $S$.  We set 
\begin{equation}
\label{eq:prior.S.given.size}
\pi_\lambda(S \mid |S|=s) = \frac{D(S)^{-\lambda/(2s)}1\{\kappa(S)<Cp^r\}}{\sum_{S:|S|=s}D(S)^{-\lambda/(2s)}1\{\kappa(S)<Cp^r\}}, \quad \lambda \in \RR, 
\end{equation}
where $\kappa(S)$ is the condition number of $X_S^\top X_S$, and $r$ and $C$ are positive constants specified to exclude models with extremely ill-conditioned $X_S^\top X_S$.  The constant $\lambda$ is an important feature of the proposed model, and will be discussed in more detail below.  Because of the dependence on $\lambda$ above, we will henceforth write $\pi_\lambda(S)$ for the prior of $S$.  

Second, as the prior for $\beta_S$, given $S$, we take 
\begin{equation}
\label{eq:prior.beta}
  (\beta_S \mid S,\lambda) \sim \nm \big(\phi\hat{\beta}_S, \sigma^2 g k_S (X_S^\top X_S)^{\lambda}\big).   
\end{equation}
Here $\hat{\beta}_S$ is the least squares estimator corresponding to configuration $S$ and design matrix $X_S$, $\phi \in (0,1)$ is a shrinkage factor to be specified, $g$ is a parameter controlling the prior spread, $(X_S^\top X_S)^{\lambda}$ is an adaptive powered correlation matrix, and 
\[ k_S=\tr\{(X_{S}^\top X_{S})^{-1}\} \big/ \tr\{(X_{S}^\top X_{S}\big)^{\lambda}\} 
\]
is a standardizing factor as in \citet{c15} designed to control for the scale corresponding to different values of $\lambda$.  This data-driven prior centering is advantageous because we can use a conjugate normal prior, which has computational benefits, but without sacrificing on theoretical posterior convergence rate properties \citep{c12}.  Let $\pi_\lambda(\beta_S \mid S)$ denote this prior density for $\beta_S$, given $S$.  

The power parameter $\lambda$ on the prior covariance matrix can encourage or discourage the inclusion of correlated predictors. When $\lambda > 0$, the prior shrinks the coefficients of correlated predictors towards each other; when $\lambda < 0$, they tend to be kept apart, with $\lambda=-1$ being the most familiar; and, finally, $\lambda=0$ implies prior independence.  Therefore, positive $\lambda$ would prefer larger models by capturing as many correlated predictors as possible, while negative $\lambda$ tends to select models with less collinearity; see \citet{c15} for additional discussion of this phenomenon.  Our data-driven choice of $\lambda$, along with that of the other tuning parameters introduced here and in the next subsection, will be discussed in Section~\ref{ss:tuning}.


\subsection{The posterior distribution}

For this standard linear regression model, the likelihood function is 
\[ L_n(\beta) = (2\pi \sigma^2)^{-n/2} e^{-\frac{1}{2\sigma^2} \|Y - X\beta\|^2}, \quad \beta \in \RR^p. \]
It would be straightforward to include $\sigma^2$ as an argument in this likelihood function, introduce a prior for $\sigma^2$, and carry out a full Bayesian analysis.  However, our intention is to use a plug-in estimator for $\sigma^2$ in what follows and, hence, we do not include the error variance as an argument in our likelihood function.   

Given a prior and a likelihood, we can combine the two via Bayes's formula to obtain a posterior distribution for $(S, \beta_S)$ or, equivalently, for the coefficient vector $\beta$.  To put a minor twist on this standard formulation, we follow a recent trend \citep[e.g.,][]{c18,grunwald.ommen.scaling, syring.martin.scaling} and use a power likelihood instead of the likelihood itself.  That is, our posterior for $(S,\beta_S)$ is defined as 
\[ \pi_\lambda^n(S, \beta_S) \propto L_n(\beta_{S+})^\alpha \pi_\lambda(\beta_S \mid S) \pi_\lambda(S), \]
where $\beta_{S+}$ is the $p$-vector obtained by filling in around $\beta_S$ with zeros in the entries corresponding to $S^c$, and $\alpha \in (0,1)$ is a regularization factor.   Then the posterior distribution for $\beta$, denoted by $\Pi_\lambda^n$, is obtained by summing over all configurations $S$, i.e.,
$$
\Pi_\lambda^n(A) \propto \int_{\{\beta_S: \beta_{S+} \in A\}} \sum_S \pi_\lambda^n(S, \beta_S) \, d\beta_S, \quad A \subseteq \RR^p.
$$
Since one of our primary objectives is variable selection, it is of interest that we can obtain a closed-form expression for the posterior distribution of $S$, up to a normalizing constant, a result of our use of a conjugate normal prior for $\beta_S$, given $S$.  That is, we can integrate out $\beta_S$ above to get a marginal likelihood for $Y$, i.e., 
\begin{align}
m_\lambda(Y \mid S) & =(2\pi\sigma^2)^{-n/2}\prod_{i=1}^s\bigl(1+\alpha g k_S d_{S,i}^{\lambda+1}\bigr)^{-1/2} \notag \\
& \qquad \times \exp\Bigl[-\frac{\alpha}{2\sigma^2} \Bigl\{ \|y-\hat{y}_S\|^2 +(1-\phi)^2 \sum_{i=1}^s\frac{d_{S,i}}{1+\alpha g k_S d_{S,i}^{\lambda+1}}\theta_{S,i}^2 \Bigr\} \Bigr], \label{eq:marginal}
\end{align}
where $d_{S,i}$ is the $i^\text{th}$ eigenvalue of $X_S^\top X_S$; also, $\Gamma_S \Lambda_S \Gamma_S^\top$ is the spectral decomposition of $X_S^\top X_S$, with $\Lambda_S=\text{diag}(d_{S,1},\ldots,d_{S,s})$, and  $\theta_{S,i}$ is the $i^\text{th}$ element of $\theta_S=\Gamma_S^\top \hat{\beta}_S$.  Then it is straightforward to get the posterior distribution for $S$:
\begin{equation}
\label{eq:marginal.S}
\pi_\lambda^n(S) \propto m_\lambda(Y \mid S) \, \pi_\lambda(S). 
\end{equation} 
The variable selection method described in Section~\ref{s:comp} and illustrated in Sections~\ref{s:sim}--\ref{s:real} is based on this posterior distribution.

\subsection{A motivating example}
\label{ss:illustration}

We now give a simple example to illustrate the effects of incorporating $\lambda$ into \eqref{eq:prior.S.given.size} and \eqref{eq:prior.beta}.  Consider a case with $n=p=5$ and let $X=X_{n\times p}$ have iid rows, each with a standard multivariate normal with first-order autoregressive dependence and correlation parameter $\rho$.  Given $X$, the conditional distribution of the response is determined by the linear model 
\[ y_i=x_{i1}+0.8x_{i2}+\eps_i, \quad \text{where} \quad \eps_1,\ldots,\eps_5 \iid \N(0, 1). \]
The black, blue, and red curves in Figure~\ref{fig:illustration} represent $\lambda \mapsto \log \pi_\lambda^n(S)$, for three different $S$ configurations, namely, the true configuration $S^\star=\{1, 2\}$, $S^-=\{1\}$, and $S^+=\{1, 2, 3\}$, respectively.  Panel~(a) corresponds to a high correlation case, $\rho = 0.8$, we see that the ECAP-based posterior would prefer $S^\star$ for suitably large $\lambda$; compare this to the choice $\lambda \equiv -1$ in \citet{c12} which would prefer the smaller configuration $S^-$.  On the other hand, when the correlation is relatively low, as in Panel~(b), we see that large positive $\lambda$ would encourage larger configuration while the true configuration would be preferred for sufficiently large negative values of $\lambda$.  The take-away message is that, by allowing $\lambda$ to vary, the ECAP-based model has the ability to adjust to the correlation structure, which can be beneficial in identifying the relevant variables.  


\begin{figure}
\centering
    \begin{subfigure}[b]{0.5\textwidth}
                \centering
                \includegraphics[width=0.9\textwidth]{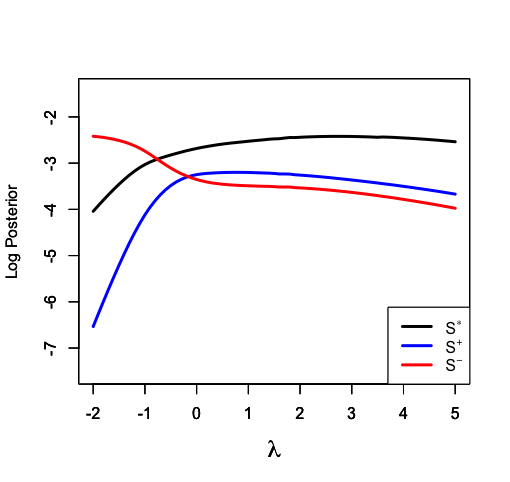}
                \caption{$\rho=0.8$}
                \label{fig:rho0.8}
    \end{subfigure}%
    \begin{subfigure}[b]{0.5\textwidth}
                    \centering
                \includegraphics[width=0.9\textwidth]{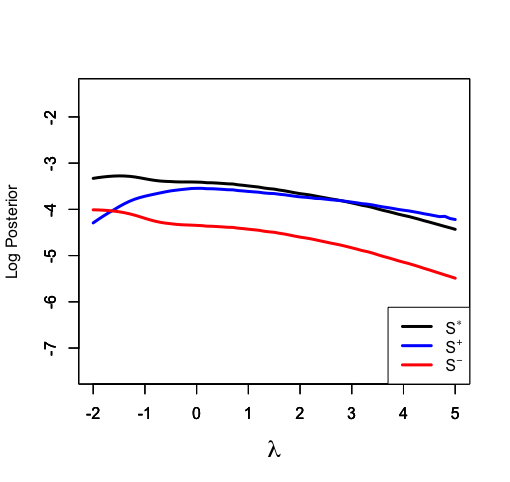}
                \caption{$\rho=0.1$}
                \label{fig:rho0.1}
    \end{subfigure}%
    \caption{Plot of $\lambda \mapsto \log\pi_\lambda^n(S)$ for three different $S$ and two different $\rho$.}
\label{fig:illustration}
\end{figure}

\section{Posterior convergence  properties}
\label{s:inf}

\subsection{Setup and assumptions}

We will stick with the standard notation given previously, but it will help to keep in mind that $Y^n=(Y_1^n, Y_2^n, ..., Y_n^n)$ and $X^n=((X_{ij}^n))$ are better understood as triangular arrays.  Therefore, we can have $p$, $s^\star$, and $R$ all depend on $n$.  We assume throughout that $s^{\star} \le R \le n \ll p$; more precise conditions given below.  We also assume that $\lambda$, $g$, and $\sigma^2$ are fixed constants in this setting, not parameters to be estimated/tuned.  Therefore, to simplify the notation here and in the proofs, we will drop the subscript $\lambda$ and simply write $\Pi^n$ for the posterior for $\beta$ instead of $\Pi_\lambda^n$.

For the fundamental problem of estimating the mean response, the minimax rate does not depend on the correlation structure in $X$, so we cannot expect any improvements in the rate by incorporating this correlation structure in our prior distribution.  
Therefore, our goal here is simply to show that the minimax rates can still be achieved while leaving room to adjust for collinearity in finite-samples.  The finite-sample benefits of the correlation-adaptive prior will be seen in the numerical results presented in Section~\ref{s:sim}.   

We start by stating the basic assumptions for all the results that follow, beginning with two simple-to-state assumptions about the asymptotic regime.  In particular, relative to $n$, the true configuration is not too complex. 

\begin{asmp}
\label{asmp:1}
$s^\star=o(n)$.
\end{asmp}



The next assumption puts a very mild size condition on the non-zero regression coefficients as well as on the user-specified shrinkage factor $\phi=\phi_n$ in the prior.  


\begin{asmp}
\label{asmp:2}
The true $\beta^\star$ satisfies $s^\star \{n \|\beta_{S^\star}^\star\|\}^{-1} \to 0$ and, moreover, the factor $\phi_n \in (0,1)$ satisfies $1-\phi_n = O(s^\star \{n \|\beta_{S^\star}^\star\|\}^{-1})$.  
\end{asmp}

The condition on $\beta^\star$ above basically says that $\min_{j \in S^\star} |\beta_j^\star| > s^{\star 1/2} n^{-1}$, which is typically much weaker than the usual {\em beta-min} condition---see \eqref{eq:beta.min} in Section~\ref{SS:selection}---that is required for variable selection consistency.  Note, however, that one would need to know $\|\beta_{S^\star}^\star\|$ to check this assumption about $\phi_n$, and we present a data-driven choice of $\phi_n$ that satisfies this condition in Section~\ref{SSS:phi}.  

Finally, we need to make some assumptions on the $n \times p$ design matrix $X$.  For a given configuration $S$, let $\lambda_{\min}(S)$ and $\lambda_{\max}(S)$ denote the smallest and the largest eigenvalues of $n^{-1} X_S^\top X_S$, respectively.  Next, define 
\[ \ell(s) = \min_{S: |S|=s} \lambda_{\min}(S) \quad \text{and} \quad u(s) = \max_{S: |S|=s} \lambda_{\max}(S). \]
Recall that these depend (implicitly) on $n$ because of the triangular array formulation.  It is also clear that $\ell(s)$ and $u(s)$ are non-increasing and non-decreasing functions of the complexity $s$, respectively.  If $\kappa(S) = \lambda_{\max}(S)/\lambda_{\min}(S)$ is the condition number of $n^{-1} X_S^\top X_S$, then we can define 
\[ \omega(s) = \max_{S: |S|=s} \kappa(S), \]
and get the relation $\omega(s) \leq u(s)/\ell(s)$.  

\begin{asmp}
\label{asmp:3}
$0 < \liminf_n \ell(s^\star) < \limsup_n u(s^\star) < \infty$.  
\end{asmp}

This assumption says, roughly, that every submatrix $X_{S}$, for $|S| \leq s^\star$, is full rank, and is implied by, for example, the sparse Riesz condition of order $s^\star$ in \citet{c22}.  More general results can surely be derived, i.e., with weaker conditions on $X$, but at the expense of more complicated statements and assumptions.

\subsection{Rates under prediction error loss}
\label{ss:prediction}

Define the set 
\begin{equation}
\label{eq:B.set}
B_\eps = \{\beta \in \RR^p: \|X\beta - X\beta^\star\|^2 > \eps\}, \quad \eps > 0, 
\end{equation}
which contains those $\beta$ that would do a relatively poor job of estimating the mean $X\beta^\star$ or, equivalently, of predicting the response $Y$.  We, therefore, hope that the posterior $\Pi_\lambda^n$ assigns asymptotically negligible mass there.  The following theorem makes this precise.  Recall the definitions of the prior and, in particular, the quantities $a$ and $r$.    

\begin{thm}
\label{thm:prediction.rate}
Under Assumptions~\ref{asmp:1}--\ref{asmp:3}, there exists a constant $M$ such that 
\[ \sup \E_{\beta^\star}\{\Pi^n(B_{M\eps_n})\} \to 0, \quad n \to \infty, \]
where the supremum is over all $\beta^\star$ such that $|S_{\beta^\star}|=s^\star$, 
\[ \eps_n = \max\{q(R, \lambda, r, a), s^\star \log(p/s^\star)\}, \]
and 
\[ q(R, \lambda, r, a) = \begin{cases} R\{r(1+\lambda) - a\}\log p & \text{if $\lambda \in [0,\infty)$} \\ R (r - a) \log p & \text{if $\lambda \in [-1,0)$} \\ R(-r\lambda - a) \log p & \text{if $\lambda \in (-\infty,-1)$}. \end{cases} \]
\end{thm}

\begin{proof} 
See Appendix~\ref{SS:proof1}.  
\end{proof}

In the so-called ordinary high-dimensional regime \citep[e.g.,][]{rigollet2012sparse}, $s^\star \log(p / s^\star)$ is the minimax concentration rate.  So the take-away message here is, as long as the $(a,r)$ in \eqref{eq:fns} and \eqref{eq:prior.S.given.size} are chosen to satisfy
\[ a > r \max\{1+\lambda, 1-\lambda\}, \]
the ECAP posterior attains the minimax optimal rate.

\subsection{Effective posterior dimension}

Theorem~\ref{thm:prediction.rate} suggests that the posterior for $\beta$ concentrates near the true $\beta^\star$ in some sense.  However, since $\beta^\star$ is sparse, it would be beneficial, for the sake of parsimony, etc, if the posterior also concentrated on a roughly $s^\star$-dimensional subset of $\RR^p$.  This does not follow immediately from the prediction loss result but, the following theorem reveals that the effective dimension of the posterior $\Pi^n$ for $\beta$ is approximately $s^\star$.  


\begin{thm}
\label{thm:dimension}
Suppose $f_n(s)$ has $(a,r)$ that satisfy $a>r\max\{1+\lambda, 1, -\lambda\}$, and define 
\begin{equation}
\label{eq:rho0}
\rho_0 = \frac{a+1}{a-r\max\{1+\lambda, 1, -\lambda\}} > 1. 
\end{equation}
Then, under Assumptions~\ref{asmp:1}--\ref{asmp:3}, we have $\sup \E_{\beta^\star}\{\Pi^n(U_n)\} \to 0$ as $n \to \infty$, where the supremum is over all $s^\star$-sparse $\beta^\star$ and $U_n = \{\beta \in \RR^p : |S_{\beta}| \ge \rho s^\star\}$ for any $\rho > \rho_0$.  
\end{thm}

\begin{proof}
See Appendix~\ref{SS:proof2}.  
\end{proof}

Aside from the economical benefits of having an effectively low-dimensional posterior distribution, Theorem~\ref{thm:dimension} is useful in the proofs of all the remaining results.  


\subsection{Rates under estimation error loss}

Define the set 
\begin{align}
\label{eq:V.set}
V_\delta = \{\beta \in \RR^p: \|\beta-\beta^\star\|^2 > \delta\}, \quad \delta > 0.
\end{align}  
Just like in Section~\ref{ss:prediction}, we hope that, for suitable $\delta$, the posterior $\Pi^n$ will assign asymptotically negligible mass to this set of undesirable $\beta$s.   

\begin{thm}
\label{thm:estimation.rate}
Suppose $f_n(s)$ has $(a,r)$ that satisfy $a>r\max\{1+\lambda, 1-\lambda\}$.  Under Assumptions~\ref{asmp:1}--\ref{asmp:3}, there exists a constant $M >0$ such that $\sup \E_{\beta^\star}\{\Pi^n(V_{M\delta_n})\} \to 0$ as $n \to \infty$, where the supremum is over all $s^\star$-sparse $\beta^\star$, 
\begin{equation}
\label{eq:delta.rate}
\delta_n = \frac{s^\star \log(p/s^\star)}{n\ell((\rho+1)s^\star)}, 
\end{equation}
and $\rho$ is greater than $\rho_0$ in \eqref{eq:rho0}.  
\end{thm}

\begin{proof}
See Appendix~\ref{SS:proof3}.  
\end{proof}

Under Assumption~\ref{asmp:1} and \ref{asmp:3}, $\ell((\rho+1)s^\star)$ is bounded with probability $1$. Hence, our rate is $n^{-1}s^\star \log(p/s^\star)$, better than the rate $n^{-1}s^\star\log p$ for lasso   \citep{zhang2008sparsity}. 

\subsection{Variable selection consistency}
\label{SS:selection}

One of our primary objectives in introducing the $\lambda$-dependent prior distribution that accounts for collinearity structure in the design matrix is for the purpose of more effective variable selection.  So it is imperative that we can demonstrate theoretically that, at least asymptotically, our posterior distribution will concentrate on the correct configuration $S^\star$.  The following theorem establishes this variable selection consistency property.  


\begin{thm}
\label{thm:selection}
In addition to Assumptions~\ref{asmp:1}--\ref{asmp:3}, assume that the constant $a$ in the prior is such that $a > 1$ and $p^a \gg s^\star e^{Gs^\star}$, where 
\[ G = \alpha + \tfrac12 \log\{1+\alpha g \kappa(S^\star)^{\max\{\lambda+1, 1, -\lambda\}}\} = O(1). \]
Then 
\[ \sup \E_{\beta^\star}\bigl[ \Pi^n(\{\beta: S_\beta \supset S_{\beta^\star}\}) \bigr] \to 0, \quad n \to \infty, \]
where the supremum is over all $\beta^\star$ that are $s^\star$-sparse.  Furthermore, if 
\begin{equation}
\label{eq:beta.min}
\min_{j \in S^\star} |\beta_j^\star| \geq \varrho_n := \frac{\sigma}{\sqrt{(\ell(s^\star))}}\Bigl\{\frac{2M}{n\alpha(1-\alpha)}\log p\Bigr\}^{1/2},
\end{equation}
where $M > a+1$ and $p^{M-(a+1)} \gg m^{s^\star}$, then 
\[ \E_{\beta^\star}\bigl[ \Pi^n(\{\beta: S_\beta \not\supseteq S_{\beta^\star}\}) \bigr] \to 0, \quad n \to \infty. \]
If both sets of conditions hold, then variable selection consistency holds, i.e., 
\[ \E_{\beta^\star}\bigl[ \Pi^n(\{\beta: S_\beta = S_{\beta^\star}\}) \bigr] \to 1, \quad n \to \infty. \]
\end{thm}

\begin{proof}
See Appendix~\ref{SS:proof4}.  
\end{proof}

The extra conditions on $(p,s^\star)$ in Theorem~\ref{thm:selection} effectively require that the true configuration size, $s^\star$, is small relative to $\log p$ and, furthermore, that the constant $a$ in \eqref{eq:fns} is large enough that $f_n(s)$ concentrates on comparatively small configurations.  Also, the non-zero $\beta^\star$ values are more difficult to detect if their magnitudes are small.  This is intuitively clear, and also shows up in our simulation results for Cases~1--2 in Section~\ref{s:sim}.  Theorem~\ref{thm:selection} gives a mathematical explanation of this intuition, stating that variable selection based on our empirical Bayes posterior will be correct asymptotically if the true signals have magnitude greater than the threshold $\varrho_n$ in \eqref{eq:beta.min}.  This is the so-called {\em beta-min} condition; see, e.g., \citet{buhlmann2011statistics} or \citet{arias2014}.

\ifthenelse{1=1}{}{
\begin{thm}
\label{thm:selection1}
For constants $a \ge 3$ and $M>a+3$, and if $p^{M-(a+3)} >> m^{s^\star}$ where  $m=(1+\alpha g \kappa(S^\star)^{\max\{\lambda+1, 1, -\lambda\}})^\frac{1}{2}$, then take any $\beta^\star$ with $S^\star=S_{\beta^\star}$, and
$$
\underset{j \in S^\star}{\min}|\beta^\star_j| \ge \varrho_n=\frac{\sigma}{l(s^\star)}\big\{\frac{2M}{\alpha(1-h\alpha)}\log p\big\}^\frac{1}{2},
$$
we have $\E{}_{\beta^\star}\big[\Pi^n(\beta:S_\beta \nsupseteq S_{\beta^\star})\big] \to 0$, uniformly over $\beta^{\star}$
\end{thm}

When the true signals are not relatively too small, i.e. $\underset{j \in S^\star}{\min}|\beta^\star_j| \ge \varrho_n$, our method can successfully identify the active predictors. And asymptotically, the selected model can cover the true model with probability $1$. 

\begin{proof}
See Appendix~\ref{S:proof4}.  
\end{proof}

\begin{thm}
\label{thm:selection2}
If $p$ and the constant $a>0$ in $f_n(s)$ are such that $p^{a} \gg s^\star m^{s^\star}$, where $m=(1+\alpha g \kappa(S^\star)^{\max\{\lambda+1, 1, -\lambda\}})^\frac{1}{2}$. Then $\E{}_{\beta^\star}\big[\Pi^n(\beta:S_\beta \supset S_{\beta^\star})\big] \to 0$, uniformly over $\beta^{\star}$
\end{thm}

\begin{proof}
See Appendix~\ref{S:proof5}.  
\end{proof}

Theorem 4 guarantees that the selected model can always contain the true covariates, while Theorem 5 claims that no unnecessary variables will be included as $n \to \infty$. Therefore, under certain conditions on $a$ and $p$, this model selection approach based on ECAP can hit the correct model consistently.
}

\section{Implementation details}
\label{s:comp}

\subsection{Stochastic search of the configuration space}

In order to compute the posterior probability, we need to evaluate $\pi_\lambda(S \mid |S|=s)$ as defined in \eqref{eq:prior.S.given.size}.  The difficulty comes from the denominator of \eqref{eq:prior.S.given.size}, i.e., 
\[ \sum_{S:|S|=s} D(S)^{-\lambda/2s}1\{\kappa(S)<Cp^r\}, \]
where, again, $D(S) = |X_S^\top X_S|$ is the determinant. Since $C$ and $r$ could be chosen large enough so that only extremely ill-conditioned cases would be excluded, there are approximately $\binom{p}{s}$ terms in the above summation. So brute-force computation can be done only for relatively small $p$.  Given that the eigenvalues are bounded from above and below for a given predictor matrix $X$, $D(S)^{1/s}$ is too, so it is not unreasonable to approximate the above summation by $\binom{p}{s}$. This approximation is exact in the case of $\lambda=0$ if all $S$ are included, and numerical experiments suggest that it is stable across a range of $p$, $s$, and $\lambda$.  Therefore, using this approximation, the posterior distribution for $S$ that we use is given by 
\begin{equation}
\label{eq:post.S}
\pi_\lambda^n(S) \propto m_\lambda(Y \mid S) \, D(S)^{-\frac{\lambda}{2|S|}} \binom{p}{|S|}^{-1} f_n(|S|). 
\end{equation}

Markov chain Monte Carlo (MCMC) methods can be used to compute this posterior but this tends to be inefficient in high-dimensional problems.  As an alternative, we employ the simplified shotgun stochastic search algorithm with screening \citep[S5,][]{c23}, a simplified version of shotgun stochastic search \citep[SSS,][]{c24}, to explore our posterior distribution. Different from traditional MCMC method, SSS does not attempt to approximate the posterior probability; instead, it only tries to explore high posterior probability regions as thoroughly as possible.         

Here is a summary of our SSS algorithm.  Let $S$ be a configuration of size $s$, with $\pi_\lambda^n(S)$, its corresponding (unnormalized) posterior. Define the neighborhood of $S$ as $\text{nbd}(S)=\{ S^+,S^0,S^-\}$, where $S^+$ is the set containing all $(s+1)$ dimensional configurations that include $S$, $S^0$ is the set containing all $s$-dimensional configurations that only have one variable different from variables in $S$, and $S^-$ is the set containing all $(s-1)$-dimensional configurations that are nested in $S$. The $t^\text{th}$ iteration of SSS goes as follows:

\begin{enumerate}
\item Given $S^{t}$, compute $\pi_\lambda^n(S)$ for all $S \in \text{nbd}(S^{t})=\{ S^{t+},S^{t0},S^{t-}\}$.           
\vspace{-1mm}
\item Sample $S^t_1$, $S^t_2$, $S^t_3$ respectively from $S^{t+}$, $S^{t0}$, $S^{t-}$, with probabilities $\propto \pi_\lambda^n(S_\cdot^t)$.
\vspace{-1mm}
\item Sample $S^{t+1}$ from $\{S^t_1,S^t_2,S^t_3\}$ with probabilities proportional to $\pi_\lambda^n(S^{t+})$, $\pi_\lambda^n(S^{t0})$, and $\pi_\lambda^n(S^{t-})$, obtained by summing.
\end{enumerate}
All visited configurations are recorded, and the final chosen configuration can be the maximum {\em a posteriori} model, median probability model (the model including those variables of which marginal inclusion probability is not less than $0.5$), or something else.  For our simulations in Section~\ref{s:sim}, the selected configuration $\hat S$ is the median probability model.  

While SSS has the ability to explore many more high posterior configurations than MCMC, it is still computationally expensive, especially in high-dimensional case.  For this reason, we adopt the S5 algorithm, which uses a screening technique to significantly decrease the computational cost.

\subsection{Choice of tuning parameters}
\label{ss:tuning}

\subsubsection{Choice of $\lambda$}

An ``ideal'' value $\lambda^\star$ of $\lambda$ is one that minimizes the Kullback--Leibler divergence of the marginal distribution $m_\lambda(y) = \sum_S m_\lambda(y \mid S) \pi_\lambda(S)$ from the true distribution of $Y$ or, equivalently, one that maximizes the expected log marginal likelihood, i.e., 
\[ \lambda^\star = \arg\max_\lambda \E\{ \log m_\lambda(Y) \}. \]
Unfortunately, the ideal value $\lambda^\star$ is not available because we do not know the true distribution of $Y$, nor can we estimate it with an empirical distribution.  However, a reasonable estimate of this ideal $\lambda$ is 
\[ \hat\lambda = \arg\max_\lambda \log m_\lambda(Y). \]
Indeed, Figure~\ref{fig:spa} shows $\log m_\lambda(Y)$ for several different $Y$ samples, along with an approximation of $\E\{\log m_\lambda(Y)\}$ based on point-wise averaging.  Notice that the individual log marginal likelihoods are maximized very close by where the expectation is maximized.  

\begin{figure}
    \centering
    \scalebox{0.08}{\includegraphics{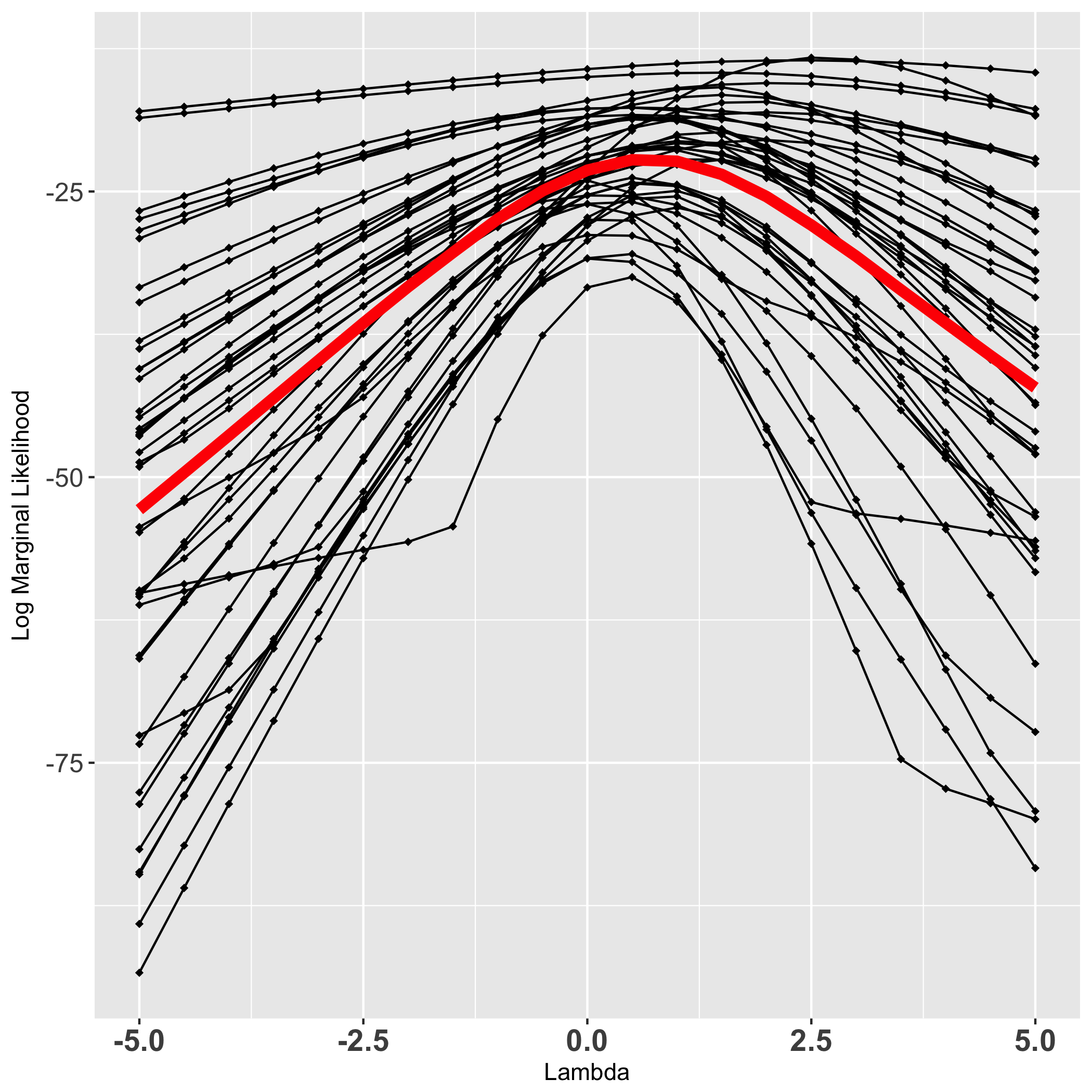}}
    \caption{Black lines are $\lambda \mapsto \log m_\lambda(Y)$ for different $Y$ samples, and the red line is the point-wise average, which approximates $\lambda \mapsto \E\{\log m_\lambda(Y)\}$.}
\label{fig:spa}
\end{figure}

There is still one more obstacle in obtaining $\hat\lambda$, namely, that we cannot directly compute the summation involved in $m_\lambda(Y)$ due to the large number of configurations $S$.  Fortunately, we can employ an importance sampling strategy to overcome this.  Specifically, we have 
\begin{align*}
m_\lambda(Y) & =\frac{\sum_S m_\lambda(Y \mid S) D(S)^{-\lambda/2|S|} f_n(|S|) \binom{p}{|S|}^{-1}}{\sum_S D(S)^{-\lambda/2|S|} f_n(|S|) \binom{p}{|S|}^{-1}}\\
&\approx \frac{\sum_{\ell=1}^N m_\lambda(Y \mid S_\ell) D(S_\ell)^{-\lambda/2|S_\ell|}}{\sum_{\ell=1}^N D(S_\ell)^{-\lambda / 2|S_\ell|}},
\end{align*}
where $\{S_\ell: \ell=1,\ldots,N\}$ are samples from $\pi_0(S)\propto f_n(|S|){p \choose |S|}^{-1}$.  In our numerical results that follow, we use this $m_\lambda(Y)$ to estimate $\hat\lambda$ defined above. 

As discussed in Section~\ref{ss:illustration}, $\lambda$ plays an important role in both model prior and coefficient prior.  That is, for a fixed size $s$, a positive $\lambda$ favors model including predictors with relatively high correlation; a negative $\lambda$ favors model including predictors with relatively low correlation; $\lambda$ equals zero actually put equal mass on each model regardless of their predictors' correlation structure. The $\lambda$ in the conditional prior for $\beta_S$, given $S$, has a similar effect; see \citet{c15}.  Thus, a ``good'' estimate of $\lambda$ should be such that it reflects the correlation structure in $X$.  

To help see this, consider a few examples, each with $X$ of dimension $n=100$ and $p=500$, having an AR(1) correlation structure with varying correlation $\rho$ and true configuration $S^\star$.  In particular, we consider two configurations:
\begin{align*}
S_1^\star & = \{11,\ldots,15,31,\ldots,35\} \\
S_2^\star & = \{1,51,100,151,200,251,300,351,400,451\}. 
\end{align*}
Figure~\ref{fig:kl} shows $\hat{\lambda}$ chosen by maximizing the marginal likelihood in three different cases, and we argue that $\hat\lambda$ is at least in the ``right direction.'' In particular, when the true predictors are highly correlated, as in Panel~(a), $\hat{\lambda}$ tends to be positive which encourages highly correlated predictors to be selected; and when the true predictors have low correlation, as in Panel~(b), the estimate of $\lambda$ is close to 0 hence a nearly uniform prior for $S$.  The situation in Panel~(c) is different since the true predictors are minimally correlated while unimportant predictors are highly correlated. In this case, $\hat{\lambda}$ tends to be negative which discourages selecting the highly correlated ones that are likely unimportant.

\begin{figure}
\centering
    \begin{subfigure}[b]{0.35\textwidth}
                \centering
                \includegraphics[width=0.7\textwidth]{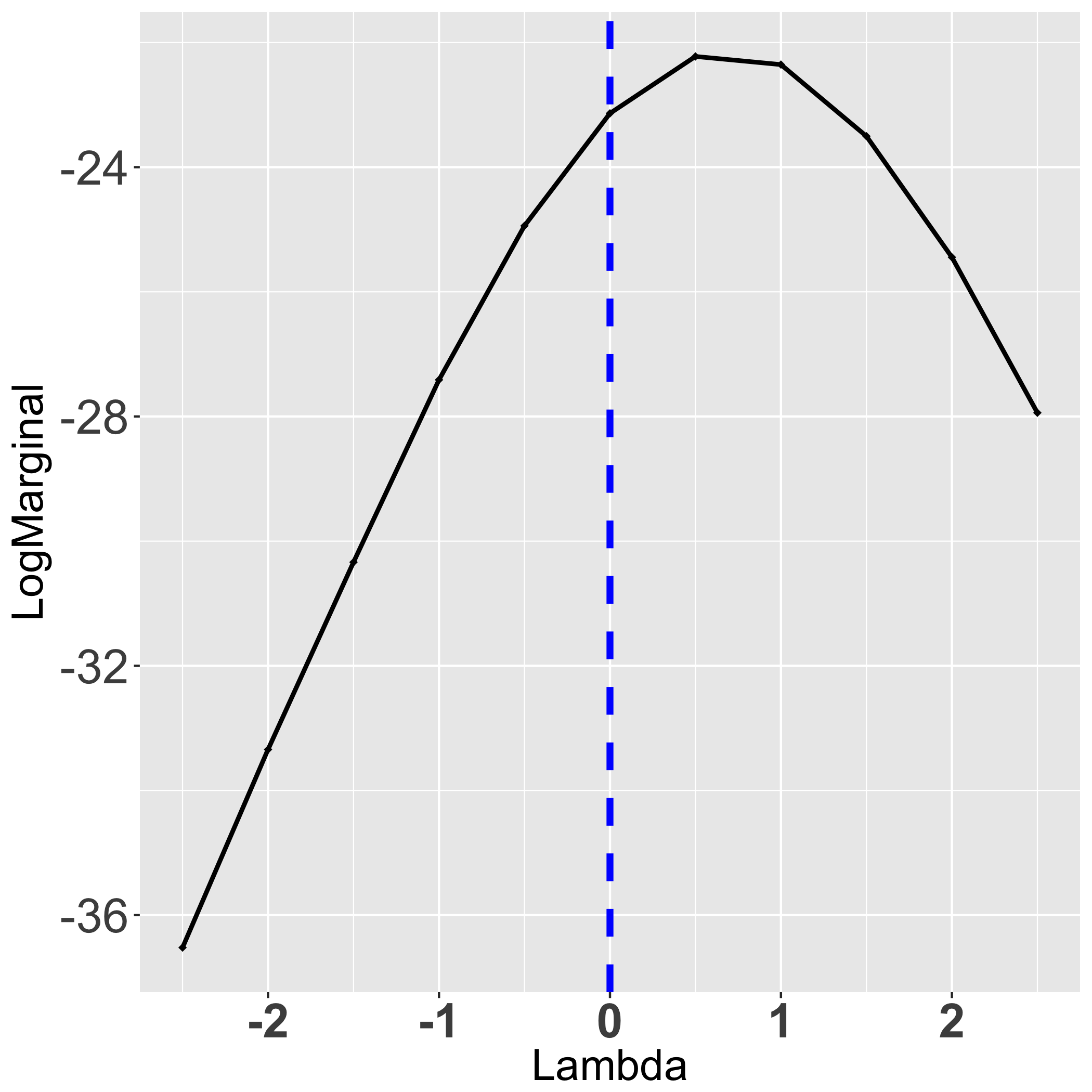}
                \caption{$\rho=0.8$, $S^\star = S_1^\star$}
                \label{fig:lampos}
    \end{subfigure}%
    \begin{subfigure}[b]{0.35\textwidth}
                    \centering
                \includegraphics[width=0.7\textwidth]{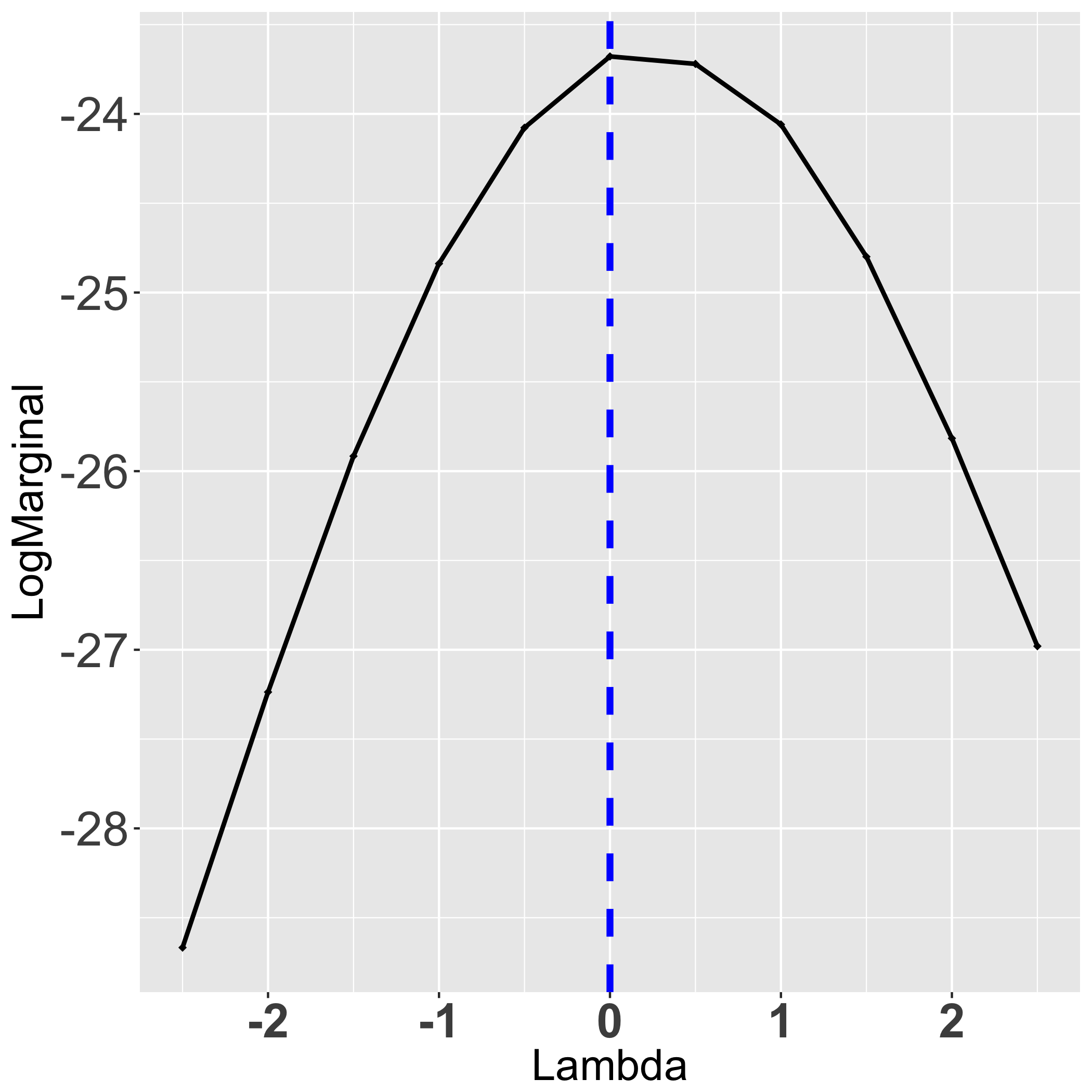}
                \caption{$\rho=0.1$, $S^\star=S_1^\star$}
                \label{fig:lam0}
    \end{subfigure}%
    \begin{subfigure}[b]{0.35\textwidth}
                    \centering
                \includegraphics[width=0.7\textwidth]{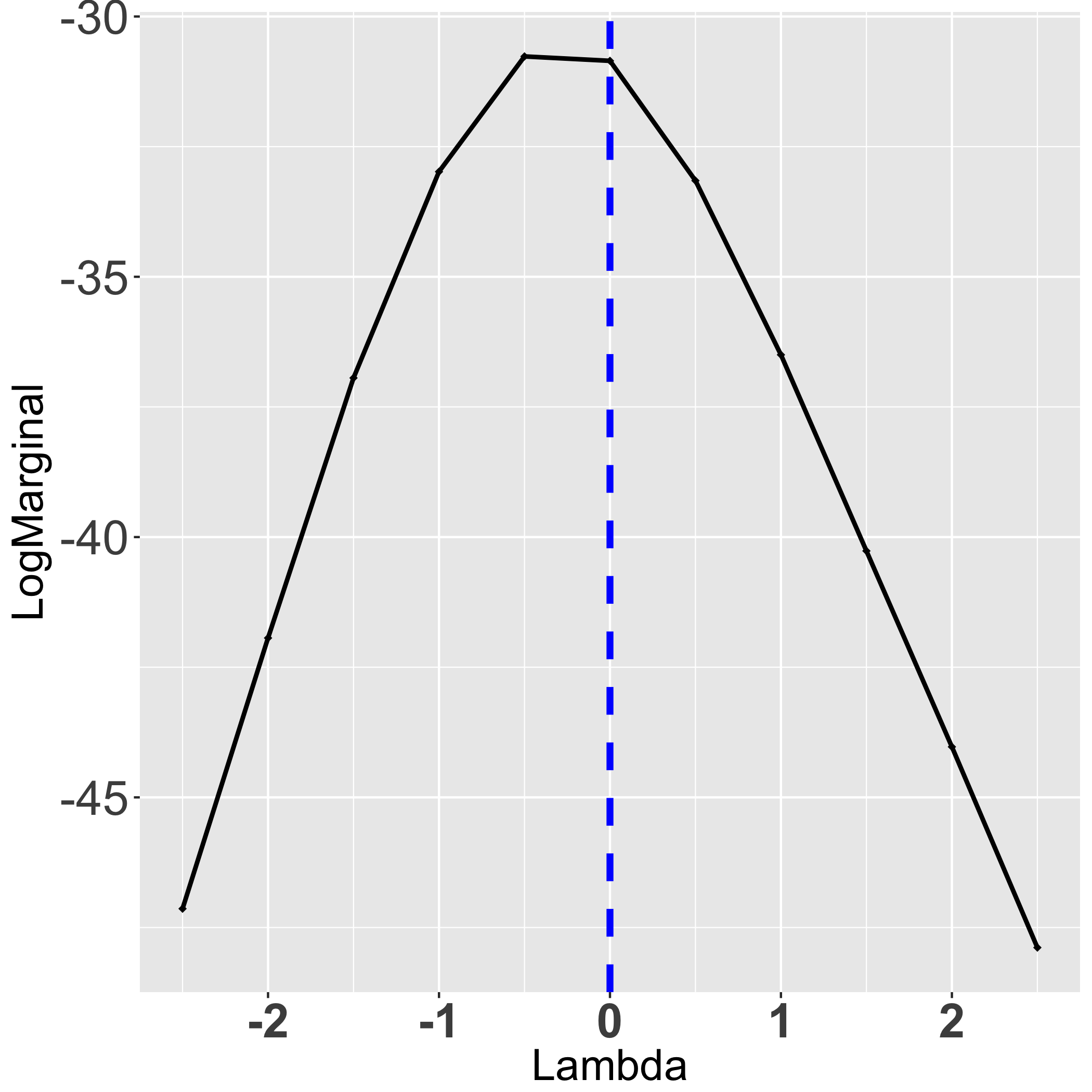}
                \caption{$\rho=0.8$, $S^\star = S_2^\star$}
                \label{fig:lamneg}
    \end{subfigure}%
    \caption{Expected log marginal likelihood versus $\lambda$, for  $\phi=0$, under different correlation structure of true configurations $S^\star$; see definition of $S_1^\star$ and $S_2^\star$ in the text.}
\label{fig:kl}
\end{figure}

\subsubsection{Choice of $g$}

Now, recall that $g$ determines the magnitude of the prior variance of $\beta_S$. If $g$ is sufficiently large, the conditional prior for $\beta_S$ is effectively flat; if $g$ is extremely tiny, then the posterior probability for $\beta_S$ will concentrate on the prior center $\phi_n\hat{\beta}_S$.  \cite{c19} proposed the unit information criteria, which amounts to taking $g=n$ in the regression setting with Zellner's prior.  \cite{c20} suggest a choice of $g=p^2$. Here, we use a local empirical Bayes estimator for $g$. That is, for given $S$ and $\lambda$, we choose a $g$ that maximizes the local marginal likelihood, that is,
\[ \hat g_S = \arg\max_g m_\lambda(y \mid S).  \]
In the special case where $\phi_n=0$ and $\lambda=-1$, and a conjugate prior for $\sigma^2$, \citet{c21} showed that $\hat g_S=\max\{F_S-1,0\}$, where $F_S$ is the usual $F$ statistic under model $S$ for testing $\beta_S=0$.  In general, our estimator, $\hat{g}_S$ must be computed numerically.

\subsubsection{Choice of $\phi$}
\label{SSS:phi}

In our choice of $\phi=\phi_n$, we seek to employ a meaningful amount of shrinkage while still maintaining the condition in Assumption~\ref{asmp:2}.  Towards this, if we view $\phi \hat\beta_{S^\star}$ as a shrinkage estimator, then it is possible to choose $\phi_n$ so that the corresponding James--Stein type estimate has smaller mean square error. In particular, a $\phi_n$ that achieves this is 
\[ \phi_n = 1-\frac{2\E\|\hat{\beta}_{S^\star}-\beta_{S^\star}^{\star}\|^2}{\|\beta_{S^\star}^{\star}\|^2 + \E\|\hat{\beta}_{S^\star}-\beta_{S^\star}^{\star}\|^2} \]
and, moreover, it can be shown that $1-\phi_n=O(s^\star\{n \|\beta_{S^\star}^\star\|\}^{-1})$.  Unfortunately, this $\phi_n$ still depends on $\beta^\star$, so we need to use some data-driven proxy for this.  We recommend first estimating $S^\star$ by $\hat S$ from the adaptive lasso, with $\hat\beta_{\hat S}$ and $\hat\sigma^2$ the corresponding least squares estimators, and then setting 
\[ \hat\phi_n = \Bigl[1-\frac{2\hat\sigma^2 \tr\{(X_{\hat S}^\top X_{\hat S})^{-1}\}}{\|\hat\beta_{\hat S}\|^2+\hat\sigma^2 \tr\{(X_{\hat S}^\top X_{\hat S})^{-1}\}} \Bigr]^+. \]
In practice, variable selection results are not sensitive to the choice of $\phi$ unless it is too close to 1.  That is, according to Figure~\ref{fig:phi}, we see good curvature in the log marginal likelihood for $\lambda$, with roughly the same maximizer, for a range of $\phi$.  The curves flatten out when $\phi$ is too close to 1, but that ``too close'' cutoff gets larger with $n$, consistent with the assumption that $1-\phi_n = O(n^{-1})$.  To ensure identifiability of $\lambda$, we manually keep our estimate of $\phi$ away from 1, in particular, we take $\tilde\phi_n = \min\{\hat\phi_n, 0.7\}$.

\begin{figure}[t]
\centering
\includegraphics[width=0.8\textwidth]{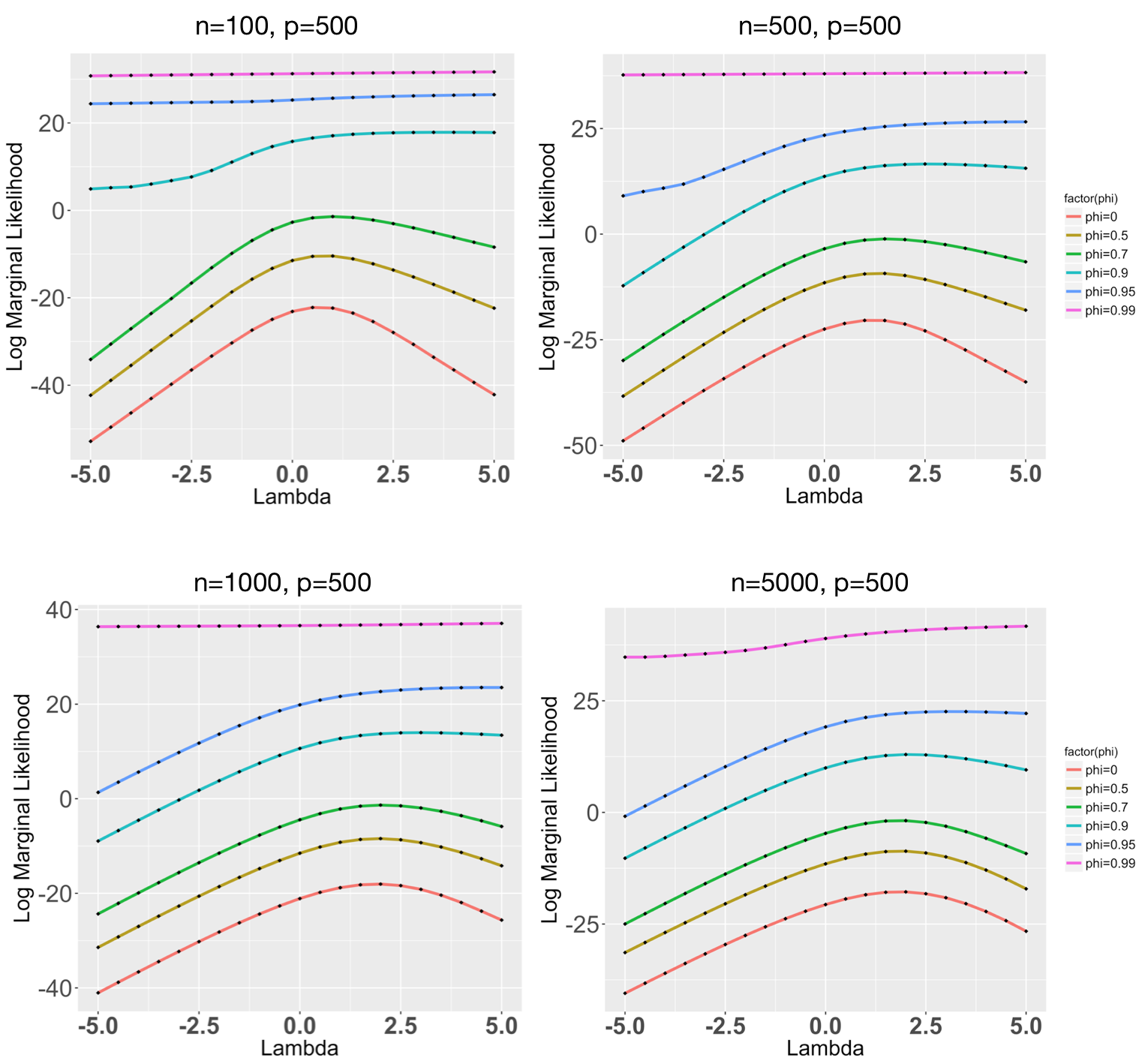}
\caption{Approximated log marginal likelihood for different values of $\phi$ with sample size $n=100, 500, 1000, 500$ and $p=500$, under Scenario~2 as is described in Section~\ref{s:sim}.}
\label{fig:phi}
\end{figure}

\subsubsection{Specification of remaining parameters}

It remains to specify the likelihood power $\alpha$, the tuning parameters $(a,c)$, specifying the prior on configuration size, and to specify a plug-in estimator for the error variance $\sigma^2$.  As in \citet{c12}, we take $\alpha=0.999$, $a=0.05$, and $c=1$.  For the error variance, we use the adaptive lasso to select a configuration and set $\hat\sigma^2$ equal to the mean square error for that selected configuration.  


\section{Simulation experiments}
\label{s:sim}

Here we investigate the variable selection performance of different methods in five simulated data settings. In each setting, $n=100$ and $p=500$ and the error variance $\sigma^2$ is set to 1. The first two settings have severe collinearity. We employ the first order autoregressive structure with $\rho=0.8$ as the covariance structure of the $n\times p$ design matrix $X$; and the true configuration $S^\star$ includes two blocks of variables that the first block contains the 11th to the 15th variable and the second block contains the 31st to the 35th variable. We explored both large and small signal cases as follows.

\begin{description}
\item[\sc Case 1:] $\beta_{S^\star} = (0.5, 0.55, 0.6, 0.65, 0.7, 0.75, 0.8, 0.85, 0.9, 0.95)^\top$ 
\item[\sc Case 2:] $\beta_{S^\star} = (1, 1.5, 2.0, 2.5, 3.0, 3.5, 4.0, 4.5, 5.0, 5.5)^\top$
\item[\sc Case 3:] In this case, we consider a block covariance setting, which is the same as the Case 4 in \cite{c11}. In this setting, interesting variables have common correlation $\rho_1=0.25$; uninteresting variables have common correlation $\rho_2=0.75$; the common correlation between interesting and uninteresting ones is $\rho_3=0.5$. The coefficients of the interesting variables are $\beta_{S^\star} = (0.6, 1.2, 1.8, 2.4, 3.0)^\top$.
\item[\sc Case 4:] This case is similar to Case 3, but let $\rho_1=0.75$, $\rho_2=0.25$, and $\rho_3=0.4$. Also, a larger $\beta_{S^\star} = (1, 1.5, 2.0, 2.5, 3.0)^\top$ is adopted.
\item[\sc Case 5:] This case is a low correlation case, which is set the same as Case 2 in \cite{c11}. All variables are set to have common correlation $\rho=0.25$ and the coefficients for interesting variables are $\beta_{S^\star} = (0.6, 1.2, 1.8, 2.4, 3.0)^\top$.
\end{description}

For each case, 1000 data sets are generated. Denoting the chosen configuration as $\hat{S}$, we compute $\prob(\hat{S}=S^\star)$ and $\prob(\hat{S}\supseteq S^\star)$ in these 1000 iterations to measure the performance of our method, denoted by ECAP.  For comparison purposes, we also consider the lasso \citep{c2}, the adaptive lasso \citep{c3}, the smoothly clipped absolute deviation \citep[SCAD,][]{c4}, the elastic net \citep[EN,][]{c26} and an empirical Bayes approach \citep[EB,][]{c12}. Tuning parameters in the first four methods are chosen by BIC.  The results are summarized in Table~\ref{table:sim}.

\begin{table}[!t]
\begin{center}
 \begin{tabular}{c c c c c}
 \hline
 Case & Method&$\prob(\hat{S}=S^\star)$&$\prob(\hat{S}\supseteq S^\star)$&Average $|\hat S|$\\
 \hline
 1 & lasso&0.082&0.996&13.61 (0.09)\\
  &alasso&0.397&0.930&10.73 (0.04)\\
   & EN&0.133&0.983 & 13.24 (0.20)\\
 & SCAD&0&0.001&12.36 (0.15)\\
 & EB&0.165&0.215 & 9.56 (0.17)\\
 & ECAP&0.263&0.342&9.65 (0.15)\\
\hline
 2 & lasso&0.297&1&11.65 (0.05)\\
  &alasso&0.356&0.412&9.33 (0.03)\\
  & EN&0.557&0.816 & 10.25 (0.07)\\
 & SCAD&0&0&7.93 (0.04)\\
 & EB&0.815&1 & 11.27 (0.91)\\
 & ECAP&0.994&1 &10.00 (0.00)\\
 \hline
3 & lasso&0&0.874&18.67 (0.12)\\
 &alasso&0.002&0.277&11.26 (0.10)\\
  & EN&0&0.945 & 19.82 (0.22)\\
 & SCAD&0.882&0.958&5.05 (0.01)\\
 & EB&0.560&0.670 & 4.69 (0.05)\\
 & ECAP&0.760&0.778 & 4.90 (0.08)\\
 \hline
4 & lasso&0.135&1&8.08 (0.09)\\
 &alasso&0.701&0.940&5.34 (0.03)\\
 & EN&0.327&0.997 & 7.33 (0.13)\\
 & SCAD&0.070&0.148&4.45 (0.04)\\
 & EB&0.793&0.822 & 4.87 (0.04)\\
 & ECAP&0.861&0.940 & 5.05 (0.07)\\
 \hline
5 & lasso&0.001&0.990&17.55 (0.15)\\
 &alasso&0.057&0.693&8.63 (0.11)\\
 & EN&0.005&0.991& 17.04 (0.28)\\
 & SCAD&0.419&0.908&5.88 (0.04)\\
 & EB&0.680&0.795 & 4.82 (0.04)\\
 & ECAP&0.827&0.919 & 4.95 (0.05)\\
 \hline
 \end{tabular}
\end{center}
\caption{Simulation results for Cases~1--5.}
\label{table:sim}
\end{table}

\ifthenelse{1=1}{}{
\begin{table}[h]
\begin{center}
 \begin{tabular}{c c c}
 \hline
 Method&$\prob(\hat{S}=S^\star)$&$\prob(\hat{S}\supseteq S^\star)$\\
 \hline
 lasso.BIC&0.308&1\\
 SCAD.BIC&0&0.003\\
 EN.BIC&0.120&1\\
 EB&0.815&1\\
 ECAP&0.994&1\\
 \hline
 \end{tabular}
\end{center}
\caption{Results for Case 2.}
\label{table:case2}
\end{table}
\begin{table}
\centering
 \begin{tabular}{c c c}
 \hline
 Method&$\prob(\hat{S}=S^\star)$&$\prob(\hat{S}\supseteq S^\star)$\\
 \hline
 lasso.BIC&0&0.015\\
 SCAD.BIC&0&0\\
 EN.BIC&0&0\\
 EB&0.560&0.670\\
 ECAP&0.760&0.778\\
 \hline
 \end{tabular}
\caption{Results for Case 3.}
\label{table:case3}
\end{table}
\begin{table}
\centering 
\begin{tabular}{c c c}
 \hline
 Method&$\prob(\hat{S}=S^\star)$&$\prob(\hat{S}\supseteq S^\star)$\\
 \hline
 lasso.BIC&0.158&0.994\\
 SCAD.BIC&0.001&0.030\\
 EN.BIC&0.003&0.864\\
 EB&0.793&0.822\\
 ECAP&0.861&0.940\\
 \hline
 \end{tabular}
\caption{Results for Case 4.}
\label{table:case4}
\end{table}
\begin{table}
\begin{center}
 \begin{tabular}{c c c}
 \hline
 Method&$\prob(\hat{S}=S^\star)$&$\prob(\hat{S}\supseteq S^\star)$\\
 \hline
 lasso.BIC&0.005&0.845\\
 SCAD.BIC&0.045&0.980\\
 EN.BIC&0.135&0.835\\
 EB&0.680&0.795\\
 ECAP&0.827&0.919\\
 \hline
 \end{tabular}
\end{center}
\caption{Results for Case 5.}
\label{table:case5}
\end{table}
}

According to these results, ECAP performs significantly better than lasso, SCAD, and EN in terms of the probability of choosing the true configuration. It also has uniformly better performance compared with EB, which is expected since the new ECAP method takes the correlation information into account. However, when considering $\prob(\hat{S}\supseteq S^\star)$, ECAP is not always the highest, e.g., Case~1.  Note that $\prob(\hat{S}=S^\star)$ and $\prob(\hat{S}\supseteq S^\star)$ for ECAP are always close to each other, which is not the case for lasso or EN. This is because the ECAP method is more likely to shrink the coefficients of unimportant predictors to zero, which is desirable if the goal is to find the true $S^\star$. 

\section{Real-data illustration}
\label{s:real}

In this section, we  examine our method in a real data example and evaluate its performance against other prevalent approaches including lasso, SCAD and the penalized credible region approach in \citet{c10}. We use the data from an experiment conducted by \cite{c29} that studies the genetics of two inbred mouse populations (B6 and BTBR). The data include $22575$ gene expressions of 31 female and 29 male mice.  Some phenotypes, including phosphoenopiruvate (PEPCK) and glycerol-3-phosphate acyltransferase (GPAT) were also measured by quatitative real-time PCR. The data are available at Gene Expression Omnibus data repository (\url{http://www.ncbi.nlm.nih.gov/geo}; accession number GSE3330).

We choose PEPCK and GPAT as response variables. Given that this is an ultra-high dimensional problem, we use marginal correlation based screening method, to screen down from $22575$ genes to $1999$ genes. Combining the screened 1999 genes with the sex variable, the final dimension of the predictor matrix is $p=2000$. After screening, we apply our method to the data and select a best subset of predictors $\hat{S}$. Then we use the posterior mean of $\beta_S$ as the estimator for $\beta$, for given $\hat{S}$ and $y$. The posterior distribution for $\beta_S$ is normal with 
\[ \text{mean} = \big(X_{\hat{S}}^\top X_{\hat{S}}+V_{\hat{S}}^{-1}\big)^{-1}\big(X_{\hat{S}}^\top y+\phi V_{\hat{S}}^{-1}\hat{\beta}_{\hat{S}}\big) \quad \text{and} \quad \text{cov} = \sigma^2\big(X_{\hat{S}}^\top X_{\hat{S}}+V_{\hat{S}}^{-1}\big)^{-1}, \]
where $V_{\hat{S}}=gk_{\hat{S}}\big(X_{\hat{S}}^\top X_{\hat{S}}\big)^{\lambda}$. For hyperparameters $\lambda$, $\phi$ and $g$, we can plug in their corresponding estimators which can be obtained in the same way as in Section~4.

In order to evaluate the performance of our approach, we randomly split the sample into a training data set of size 55 and a test set of 5. First, we apply our variable selection method to the training set and obtain the selected variables. Then conditioning on this model, we estimate the regression coefficients using the above method. Based on the estimated regression coefficient, we predict the remaining 5 observations and calculate the  prediction loss. This process is repeated 100 times, and an estimated mean square prediction error (MSPE) along with its standard error can be computed; see Table~\ref{table:real}.

\begin{table}[t]
\begin{center}
 \begin{tabular}{cc cc c}
 \hline
  &
      \multicolumn{2}{c}{PEPCK} &
      \multicolumn{2}{c}{GPAT}\\
 Method &MSPE&Model Size&MSPE&Model Size\\
 \hline
 ECAP ($p=2000$)&1.02 (0.07)& 5.04 (0.19)&2.26 (0.18)&8.34 (0.33)\\
 lasso ($p=2000$)&3.03 (0.19)&7.70 (0.96)&5.03 (0.42)&3.30 (0.79)\\
 BCR.joint ($p=2000$)&2.03 (0.14)&9.60 (0.46)&3.83 (0.34)&4.20 (0.43)\\
 BCR.marginal ($p=2000$)& 1.84 (0.14)&23.3 (0.67)& 5.33 (0.41)& 21.8 (0.72)\\
 SIS+SCAD ($p=22575$)&2.82 (0.18)&2.30 (0.09)&5.88 (0.44)& 2.60 (0.10)\\
 ECAP ($p=22575$)&0.72 (0.07)& 4.93 (0.30)&1.66 (0.52)&7.92 (0.73)\\
 \hline
 \end{tabular}
\end{center}
\caption{Mean square prediction error (MSPE) and average configuration size in the real example of Section~\ref{s:real}; numbers in parentheses are standard errors. Simulation results except for ECAP are from \cite{c10}}
\label{table:real}
\end{table}
 
In Table~\ref{table:real}, BCR.joint and BCR.marginal denote methods using joint credible sets and marginal credible sets respectively, for details, see \cite{c10}. The first four methods, ECAP, lasso, BCR.joint and BCR.marginal are implemented on the screened data with dimension $p=2000$. The fifth row is the results of sure independence screening (SIS) combined with SCAD applied to the full data $p=22575$ and the last row is the outcomes from directly applying ECAP to unscreened dataset. The stopping rules of lasso, SCAD, BCR.joint and BCR.marginal are based on BIC.

In terms of MSPE, ECAP outperforms all the other methods significantly in both PEPCK and GPAT cases, given the estimated standard errors. Moreover, the MSPE from ECAP is even smaller for the full dataset compared with the screened data. And for the model size, on average, ECAP, lasso, BCR.joint and SIS+SCAD select models with comparable sizes while BCR.marginal always chooses larger models.  Overall, ECAP performs very well in this real data example compared to these other methods in terms of both MSPE and model size.


\ifthenelse{1=1}{}{
\section{Conclusion}
\label{s:discuss}

In this paper, we propose a new empirical Bayes method that takes advantage of data's correlation structure. We prove that this new ECAP method maintains good properties in both estimation and model selection under standard assumptions. In addition to maintaining optimal posterior convergence rates, our simulated and real data examples demonstrate that the adoption of this correlation-adaptive prior leads to improve performance, especially in high collinearity cases. 
}

\section*{Acknowledgments}

The work presented herein is partially supported by the U.S.~National Science Foundation, grant DMS--1737933.

\appendix

\section{Technical details}
\label{S:proofs}

\subsection{Preliminary lemmas}

Before getting to the proofs of the main theorems, we need to suitably lower-bound the posterior denominator and upper-bound the posterior numerator, the latter depending on the type of neighborhood being considered.  In particular, for a generic measurable subset $A$ of the parameter space, write 
\[ \Pi^n(A) = \frac{N_n(A)}{D_n} = \frac{\int_A \sum_S \pi(S) R_n(\beta_{S+},\beta^*)^\alpha \pi_\lambda(\beta_S \mid S) \, d\beta_S}{\int \sum_S \pi(S) R_n(\beta_{S+},\beta^*)^\alpha \pi_\lambda(\beta_S \mid S) \, d\beta_S}, \]
where $R_n(\beta_{S+},\beta^*) = L_n(\beta_{S+}) / L_n(\beta^\star)$ is the likelihood ratio, with $\beta_{S+}$ the $p$-vector corresponding to $\beta_S$ with zeros filled in around the indices in $S$.  Lemma~\ref{lem:denominator} gives a general lower bound on the denominator $D_n$.  

In what follows, $\beta^\star$ will denote the true and sparse coefficient vector in $\RR^p$, with $S^\star = S_{\beta^\star}$ the corresponding configuration of complexity $s^\star = |S^\star|$.  

\begin{lem}
\label{lem:denominator}
Given $(\alpha, g, \lambda, \sigma^2)$, define 
\[ 
c = c(\alpha, g, \lambda, \sigma^2) = 
\begin{cases} 
\frac{1}{2}\log\{1+\alpha g \kappa(S^\star)^{1+\lambda}\}, & \text{if $\lambda \in [0, \infty)$} \\
\frac{1}{2}\log\{1+\alpha g\kappa(S^\star)\}, & \text{if $\lambda \in [-1, 0)$} \\
\frac{1}{2}\log\{1+\alpha g\kappa(S^\star)^{-\lambda}\}, & \text{if $\lambda \in (-\infty, -1)$}, 
\end{cases}
\] 
where $\kappa(S^\star)$ is the condition number of $X_{S^\star}^\top X_{S^\star}$.  Then under Assumptions 1--4,  the denominator $D_n$ satisfies $\prob_{\beta^\star}\{D_n < \pi(S^\star) e^{-c s^\star}\} \to 0$ as $n \to \infty$.  
\end{lem}

\begin{proof}
Since $D_n$ is a sum of non-negative terms, we get the trivial bound 
\begin{align*}
D_n & > \pi(S^\star) \int R_n(\beta_{S^\star},\beta^{\star}_{S^\star})^{\alpha} \N\big(d\beta_{S^\star} \mid \phi \hat{\beta}_{S^{\star}}, \sigma^2gk_{S^\star}(X_{S^\star}^\top X_{S^\star})^{\lambda} \big)  \\
& = \pi(S^\star)\prod_{i=1}^{s^\star}(1+\alpha g k_{S^\star}d_{S^\star,i}^{1+\lambda})^{-1/2}\exp\Bigl\{\frac{\alpha}{2\sigma^2}(A_n-B_n)\Bigr\},
\end{align*}
where 
\begin{align*}
A_n = n (\hat{\beta}_{S^\star}-\beta^\star_{S^\star})^\top \bigl(n^{-1} X_{S^\star}^\top X_{S^\star}\bigr)  (\hat{\beta}_{S^\star}-\beta^\star_{S^\star}), \quad B_n = (1-\phi)^2 \hat{\beta}_{S^\star}^\top Q_{S^\star}^{-1}\hat{\beta}_{S^\star}
\end{align*}
are both non-negative, and the $Q_S$ matrix is defined as 
\begin{equation}
\label{eq:Q.matrix}
Q_S = (X_{S}^\top  X_{S})^{-1}+\alpha g k_{S} (X_{S}^\top X_{S})^{\lambda}, \quad S \subseteq \{1,2,\ldots,p\}.
\end{equation}
From the well-known sampling distribution of $\hat\beta_{S^\star}$, we have 
\[ A_n \sim \sigma^2 \, \chisq(s^\star) = O_p(s^\star). \]
Next, for $B_n$, under Assumption~\ref{asmp:3}, it can be verified that the maximal eigenvalue of $Q_{S^\star}^{-1}$ is $O(n)$.  Therefore, $B_n \lesssim n (1-\phi)^2 \|\hat\beta_{S^\star}\|^2$.  
Since 
\[ \|\hat\beta_{S^\star}\|^2 \lesssim \|\hat\beta_{S^\star} - \beta_{S^\star}^\star\|^2 + \|\beta_{S^\star}^\star\|^2, \]
and the first term is $O_p(s^{\star}/n)$, it follows from Assumptions~\ref{asmp:1} and \ref{asmp:2} that $B_n = o_p(s^\star)$.  
This implies that $\prob_{\beta^\star}(A_n \leq B_n) \to 0$ which, in turn, implies that the exponential term in the lower bound for $D_n$ is no smaller than 1.  Therefore, 
\[ D_n > \pi(S^\star)\prod_{i=1}^{s^\star}(1+\alpha g k_{S^\star}d_{S^\star,i}^{1+\lambda})^{-1/2} \]
and the product can be lower-bounded by $e^{-c s^\star}$ for $c$ as defined above.  Finally, we have that $D_n \leq \pi(S^\star) e^{-cs^\star}$ with vanishing $\prob_{\beta^\star}$-probability, as was to be shown.  
\end{proof}

Next is an upper bound on the posterior numerator, $N_n = N_n(B_{\eps_n})$, where $B_{\eps_n} \subset \RR^p$ is the complement of a neighborhood like that in \eqref{eq:B.set}, relevant to Theorem~\ref{thm:prediction.rate}.  

\begin{lem}
\label{lem:numerator}
There exists a constant $d=d(\alpha,\sigma^2)$ such that 
\[ \sup_{\beta^\star} \E_{\beta^\star}(N_n) \leq e^{-d \eps_n} \sum_{S: |S| \leq R} \psi(|S|)^{|S|} \pi_\lambda(S), \]
where
\[ 
\psi(s)^2 = \begin{cases}
\omega(s)^{2(\lambda+1)}\bigl[1+\frac{q\phi^2}{g}\omega(s)^{-(\lambda+1)}\bigr]^{1-\frac{1}{q}} & \text{if $\lambda \in [0,\infty)$} \\
\omega(s)^{2}\bigl[1+\frac{q\phi^2}{g}\omega(s)^{-1}\bigr]^{1-\frac{1}{q}} & \text{if $\lambda \in [-1, 0)$} \\
\omega(s)^{-2\lambda}\bigl[1+\frac{q\phi^2}{g}\omega(s)^{\lambda}\bigr]^{1-\frac{1}{q}} & \text{if $\lambda \in (-\infty, -1)$}, 
\end{cases}
\]
where $q = (h-1)/h$ and $h > 1$ depends only on $\alpha$.  
\end{lem}

\begin{proof}
\label{SS:proof0}
Let us consider the expectation of $N_n$, given the true distribution of y, i.e. $y \sim \N(X\beta^\star, \sigma^2I)$.
Then by H{\"o}lder's inequality, for constants $h>1$ and $q=\frac{h-1}{h}$, we can find an upper bound for $\E_{\beta^\star}(N_n)$, 
\begin{align}
\label{eq:N}
\E_{\beta^\star}(N_n) \le \sum_S\pi(S)\int_{B_{\epsilon_n}}J_n(\beta_S)^{\frac{1}{h}} K_n(\beta_S)^{\frac{1}{q}} d\beta_S,
\end{align}
where
\begin{align*}
J_n(\beta_S) & =\E_{\beta^\star}\Bigl[\Bigl\{\frac{\N(y \mid X_S\beta_{S},\sigma^2I)}{\N(y \mid X_{S^\star}\beta^\star_{S^\star},\sigma^2I)}\Bigr\}^{h\alpha}\Bigr]  \\
K_n(\beta_S) & = \E_{\beta^\star}\bigl[\N^q(\beta_S \mid \phi\hat{\beta}_S,\sigma^2gk_S(X_S^\top X_S)^{\lambda})\bigr].
\end{align*}
If $h\alpha<1$, then the Renyi divergence formula gives that
    \begin{eqnarray}
       \label{eq:J} J_n(\beta_S)=e^{-\frac{\alpha(1-h\alpha)}{2\sigma^2}\|X\beta_{S+}-X\beta^\star\|^2} \leq e^{-[\alpha(1-h\alpha)/2\sigma^2]\epsilon_n}, \quad \forall \; \beta_{S+} \in B_{\eps_n},
     \end{eqnarray}
where the $p$-dimensional vector $\beta_{S+}$ is an augmented form of $\beta_S$ with  $\beta_{S+, i}=0$ if $i \notin S$.  Next, for $K_n$, after factoring out the non-stochastic terms in the multivariate normal density, there is an expectation of exponential quadratic form  to be dealt with, i.e.,
    \begin{align}
    \label{eq:Z}\E_{\beta^\star}\Big[\exp\Big\{-\frac{q}{2\sigma^2gk_S}Z\Big\}\Big],
    \end{align}
    where $Z=(\beta_S-\phi\hat{\beta}_S)^\top (X_S^\top X_S)^{-\lambda}(\beta_S-\phi\hat{\beta}_S)$ and $\hat{\beta}_S$ is the least square estimator under configuration $S$. Let $\beta_S^\star$ denote the mean of $\hat{\beta}_S$, then $\beta^\star_S=(X_{S}^\top X_{S})^{-1}X_S^TX_{S^\star}\beta_{S^\star}^\star$. Applying a spectral decomposition on $X_S^\top X_S$ in $Z$, we have,
    \[Z/(\sigma^2\phi^2)=\nu_S^\top \Lambda_S^{-(\lambda+1)}\nu_S,\]
    where $\nu_S=\Lambda_S^{1/2}\Gamma_S^\top (\beta_S-\phi\hat{\beta}_S)/(\sigma\phi)$. $\Lambda_S$ is a diagonal matrix whose diagonal elements are the corresponding eigenvalues, and $\Gamma_S$ is a matrix with columns being corresponding eigenvectors. It is not difficult to show, 
    \[\nu_S \sim \N\Big(\Lambda_S^{1/2}\Gamma_S^\top (\beta_S-\phi\beta^\star_S)/(\sigma\phi), I\Big),\]
    which implies $\nu_{S,i}$ are iid $\N(d_{S,i}^{1/2}\Gamma_{S,i}^\top (\beta_S-\phi\beta^\star_S)/(\sigma\phi), 1)$, where $d_{S,i}$ is the $i^\text{th}$ eigenvalue of $X_S^\top X_S$ and $\Gamma_{S,i}$ is the $i^\text{th}$ eigenvector.  Hence, $\nu_{S,i}^2$ has a non-central $\chisq(1)$ distribution with  non-centrality parameter  $\mu_{S,i}=\frac{1}{\sigma^2\phi^2}(\beta_S-\phi\beta^\star_S)^\top \Gamma_{S,i}d_{S,i}\Gamma_{S,i}^\top (\beta_S-\phi\beta^\star_S)$. By taking advantage of the independence of  $\nu_{S,i}^2$s and  using the moment generating function of the non-central chi-square distribution, \eqref{eq:Z} can be written as, 
    \begin{align}
    \label{eq:Zsolution}
\prod_{i=1}^s(1-2t_{S,i})^{-\frac{1}{2q}}\exp\Bigl\{\frac{\mu_{S,i}}{2q} \Bigl(\frac{1}{1-2t_{S,i}}-1 \Bigr)\Bigr\},
    \end{align}
    where $t_{S,i}=-q\phi^2 d_{S,i}^{-(1+\lambda)}/2gk_S < 1/2$.  It is clear that \eqref{eq:Zsolution} is a non-decreasing function with respect to $t_{S,i}$. And when $\lambda \ge 0$, $t_{S,i} \le t_S:= -\frac{\phi^2q}{2g}\omega(s)^{-(\lambda+1)}$. Then  $\forall \beta_S \in \mathbb{R}^{|S|}$, we can obtain an upper bound for \eqref{eq:Z} by replacing $t_{S,i}$ with $t_S$ in \eqref{eq:Zsolution},
    \begin{align}
       \label{eq:Zbound} U(\beta_S)=(1-2t_S)^{-\frac{s}{2q}}\exp\Bigl\{\Bigl(\frac{2t_S}{1-2t_S}\Bigr)\frac{1}{2q\sigma^2\phi^2}(\beta_S-\phi\beta^\star_S)^\top(X_S^\top X_S)(\beta_S-\phi\beta^\star_S)\Bigr\}
    \end{align}
For $\lambda<-1$ and $-1 \le \lambda <0$, we can get the same expression of $U_n(\beta_S)$ with $t_S=-\phi^2q \omega(s)/2g$ and $t_S=-\phi^2q \omega(s)^\lambda/2g$ respectively.  Therefore, 
    \begin{align}
        \label{eq:K}
        K_n(\beta_S) \le (2\pi\sigma^2gk_S)^{-\frac{s}{2}}D(S)^{-\frac{\lambda}{2}} U_n(\beta_S)
    \end{align}
  Since $J_n(\beta_S)$ and $K_n(\beta_S)$ are non-negative, we upper-bound the integral in \eqref{eq:N}  by 
  \begin{align}
  \label{eq:NR}
      \int_{\RR^{|S|}}&J_n(\beta_S)^{\frac{1}{h}}K_n(\beta_S)^{\frac{1}{q}}d\beta_S
    \end{align}
By plugging \eqref{eq:K} and \eqref{eq:J} into \eqref{eq:NR}, and integrating out $\beta_S$ , we bound  $\E_{\beta^\star}(N_n)$ by,
\[e^{-[\alpha(1-h\alpha)/2\sigma^2]\epsilon_n}\sum_{S:|S|\le R}\psi(|S|)^s\pi(S)\]
where function $\psi(s)$ is defined as above.
\end{proof}

\subsection{Proof of Theorem~\ref{thm:prediction.rate}}
\label{SS:proof1}

With Lemmas~\ref{lem:denominator} and \ref{lem:numerator}, the expectation of posterior probability of event $B_n = B_{M\eps_n}$, which is defined in \eqref{eq:B.set}, can be bounded by,
\begin{equation*}
\E_{\beta^\star}\bigl\{\Pi^n(B_n)\bigr\} \le e^{cs^\star-dM\eps_n}\frac{\sum_{S: s \le R}\psi(s)^s \pi(S)}{\pi(S^\star)}=e^{cs^\star-dM\eps_n}\frac{\sum_{s=1}^R\psi(s)^sf_n(s)}{\pi(S^\star)}.
\end{equation*}
While for the prior, if $\lambda \ge 0$, we can get,
$$\pi(S^\star ) \ge \omega(s^\star)^{-\lambda/2}f_n(s^\star){p \choose s^\star}^{-1}.$$  
Therefore the upper bound for the  posterior probability can be written as,
\begin{equation*}
\label{eq:posterior.B}
\begin{aligned}
   \E_{\beta^\star}\bigl\{\Pi^n(B_n)\bigr\}  
   \le e^{cs^\star-dM\eps_n}\xi_n,
\end{aligned}
\end{equation*}
where 
\[\xi_n=\frac{\omega(s^\star)^{\lambda/2}{p \choose s^\star}}{f_n(s^\star)}\sum_{s=1}^R\psi(R)^sf_n(s).\] 
Taking logarithm on both sides, we get
\begin{equation}
\log\E_{\beta^\star}\bigl\{\Pi^n(B_n)\bigr\} \le \Bigl(\frac{cs^\star}{\eps_n}-Md+\frac{\log\xi_n}{\eps_n} \Bigr)\eps_n.
\end{equation}
We require $\eps_n$ to have a certain rate such that the upper bound for posterior probability can vanish. A preliminary requirement for $\eps_n$ is $s^\star/\eps_n \to 0$, in order make $e^{cs^\star-d\eps_n}$ as $o(1)$. In addition, $\eps_n$ should satisfy $\log\xi_n=O(\eps_n)$.   Therefore, as $n \to \infty$, $cs^\star/\eps_n \to 0$ and $\log(\xi_n)/\eps_n  \to K$. Thus, for any $M$ satisfying $Md>K$, we will have \[ \log\E_{\beta^\star}\bigl\{\Pi^n(B_n)\bigr\} \to -\infty.\]

Next, we establish the rate for $\log\xi_n$. Under Assumption~\ref{asmp:1} and \ref{asmp:2}, $\omega(s^\star)$ is bounded with probability 1. By Stirling’s formula, we have that 
\[\log{p \choose s^\star}\le s^\star\log(p/s^\star)\{1+o(1)\}.\] 
Given that $f_n(s)=c^{-s}p^{-as}$, we can also have 
\[-\log f_n(s^\star) \le s^\star\log(cs^\star)+as^\star\log(p/s^\star)
=O(s^\star\log(p/s^\star)).\]
Since we have ruled out cases with extremely ill-conditioned $X_S^\top X_S$, $\omega(s)$ is bounded above by $Cp^r$. Thus, for the nonnegative $\lambda$ case, 
\[\sum_{s=1}^R \psi(R)^s f_n(s) \lesssim p^{R\big[r(1+\lambda)-a\big]}.\] 
Therefore, when $\lambda \ge 0$, the rate of $\eps_n$ should be 
\[\max\{R\big[r(1+\lambda)-a\big]\log p,   s^\star\log(p/s^\star)  \}.\] 

The proofs for $\lambda<0$ are similar. Therefore, the rate $\eps_n$ then can be rewritten as
\[ \eps_n = \max\{q(R, \lambda, r, a), s^\star \log(p/s^\star)\}, \]
where function $q$ has already been defined in Theorem~\ref{thm:prediction.rate}.

\subsection{Proof of Theorem~\ref{thm:dimension}}
\label{SS:proof2}

Followed by Lemma~\ref{lem:numerator} and Theorem~\ref{thm:prediction.rate}, we can get
    \begin{equation*}
    \E_{\beta^\star}^{1/h}\Big[\Big\{\frac{\nm(y \mid X_S\beta_{S+},\sigma^2I)}{\nm(y \mid X_S\beta_{S^\star},\sigma^2I)}\Big\}^{h\alpha}\Big]= \exp\big\{-\tfrac{\alpha(1-h\alpha)}{2\sigma^2}\|X\beta_{S+}-X\beta_{S^\star}\|^2\big\} \le 1.
    \end{equation*}
Then it is not difficult to show that,
 $$\E_{\beta^\star}\{N_n(U_n)\} \le \sum_{s=\rho s^\star}^R \psi(s)^sf_n(s). $$
With the help of Lemma~\ref{lem:denominator}, the posterior probability of event $U_n$ can be bounded as,
    \begin{equation}
    \label{eq:U.posterior}
    \E_{\beta^\star}\{\Pi^n(U_n)\} \le e^{cs^\star}\frac{\omega(s^\star)^{\lambda/2}{p \choose s^\star}}{f_n(s^\star)}\sum_{s=\rho s^\star}^R\psi(s)^sf_n(s)
    \end{equation}
    From Theorem~\ref{thm:prediction.rate}, we have 
\[\log\frac{\omega(s^\star)^{\lambda/2}{p \choose s^\star}}{f_n(s^\star)} \le (a+1+o(1))s^\star\log(p/s^\star). \] 
In addition, for $\lambda \ge 0$, if $a>r(1+\lambda)$, we can get 
    \[\sum_{s=\rho s^\star}^R\psi(s)^sf_n(s) \lesssim \exp\{-\rho s^\star [a-r(1+\lambda)]\log p\}.\] 
When $\lambda \ge 0$, $\rho > \rho_0= (a+1)\{ a-r(1+\lambda)\}^{-1}$, then $\sum_{s=\rho s^\star}^R\psi(s)^sf_n(s)$ dominates the other two terms in \eqref{eq:U.posterior}. Therefore, $\E_{\beta^\star}\{\Pi^n(U_n)\}$ will vanish as $n \to \infty$. Similarly for $\lambda <0 $, we can get the same result if $\rho > \rho_0$.

\subsection{Proof of Theorem~\ref{thm:estimation.rate}}
\label{SS:proof3}

Let $|S_{\beta-\beta^\star}|$ be the number of non-zero entries of $(\beta-\beta^\star)$. Then,
\[\|X(\beta-\beta^\star)\|^2 > n \ell(|S_{\beta-\beta^\star}|)\|\beta-\beta^\star\|^2.\] 
If $a>\max\{1+\lambda, 1, 1-\lambda\}$ and $\rho > \rho_0$ in \eqref{eq:rho0}, by Theorem~\ref{thm:dimension} and monotonicity of $\ell(s)$, we can get
\[\E_{\beta^\star}\bigl[\Pi^n(\{\beta: \ell(|S_{\beta-\beta^\star}|) \ge \ell((\rho+1)s^\star)\}) \Bigr] \to 1. \]
If we set $\delta_n$ as in \eqref{eq:delta.rate} and use Theorem~\ref{thm:prediction.rate}, we get 
\[ \E_{\beta^\star}\{\Pi^n(V_{M\delta_n}) \} \le \E_{\beta^\star} \{ \Pi^n(B_{M\ell((\rho+s^\star))\delta_n})\} \to 0, \]
where $B$ and $V$ are defined in \eqref{eq:B.set} and \eqref{eq:V.set}.

\subsection{Proof of Theorem~\ref{thm:selection}}
\label{SS:proof4}
We segment the proof\footnote{Some of the arguments presented here refer to \cite{c12}, but there are some oversights in the selection consistency results presented in the published version of that paper.  The version available at {\tt arXiv:1406.7718} contains corrections of those arguments.} of Theorem~\ref{thm:selection} into two parts. First, under Assumptions~\ref{asmp:1} and \ref{asmp:3}, we aim to show that, $\E_{\beta^\star}[\Pi^n(\{\beta: S_\beta \supset S_{\beta^\star}\})] \to 0$.  Second, under the beta-min condition, we prove that $\E_{\beta^\star}[ \Pi^n(\{\beta: S_\beta \not\supseteq S_{\beta^\star}\})] \to 0$.  We only consider positive $\lambda$ case here.  The proofs in Parts~1 and 2 below can go through the same way when $\lambda < 0$.

\subsubsection*{Part 1}

Let $S$ be any configuration containing but not equal to the true model $S^\star$, i.e. $S \supset S^\star$. Then the posterior for $S$ can be written as,
\begin{align*}
\pi^n(S) & \le \Pi^n(S) \big/ \Pi^n(S^\star) \\
& =F_SR_S
\exp\big[-\tfrac{\alpha}{2\sigma^2}\bigl\{ y^\top(P_{S\star}-P_S)y+(1-\phi)^2\hat{\beta}_S^\top Q_S^{-1}\hat{\beta}_S-(1-\phi)^2\hat{\beta}_{S^\star}^\top Q_{S^\star}^{-1}\hat{\beta}_{S^\star}\big\}\bigr] \\
& \leq F_S R_S  \exp\bigl[ \tfrac{\alpha}{2\sigma^2}\bigl\{ y^\top (P_{S}-P_{S^\star})y+(1-\phi)^2\hat{\beta}_{S^\star}^\top Q_{S^\star}^{-1}\hat{\beta}_{S^\star}\bigr\}\bigr],
\end{align*}
where $F_S = \pi(S) / \pi(S^\star)$, $Q_S$ is as in \eqref{eq:Q.matrix}, and 
\begin{align*}
R_S =\frac{\prod_{i=1}^s(1+\alpha gk_Sd_{S,i}^{\lambda+1})^{-\frac{1}{2}}}{\prod_{i=1}^{s^\star}(1+\alpha gk_{S^\star}d_{S^\star,i}^{\lambda+1})^{-\frac{1}{2}}}.
\end{align*}
Applying H{\"o}lder's inequality with the same constants in \eqref{eq:N}, $h>1$ and $q=(h-1)/h$, we get that $\E_{\beta^\star}\{\pi^n(S)\}$ is upper-bounded by 
\begin{equation}
\label{eq:S.posterior}
F_S R_S\E_{\beta^\star}^{1/h}\bigl[\exp\{\tfrac{\alpha h}{2\sigma^2} y^\top(P_S-P_{S^\star})y\}\bigr] \E_{\beta^\star}^{1/q}\bigl[\exp\{\tfrac{\alpha q(1-\phi)^2}{2\sigma^2}\hat{\beta}_{S^\star}^\top Q_{S^\star}^{-1}\hat{\beta}_{S^\star}\}\bigr].
\end{equation}
First, given that $S\supset S^{\star}$, $(P_S-P_S^\star)$ is idempotent and $(P_S-P_S^\star)X_{S^\star}=0$. Therefore, $y^\top (P_S-P_{S^\star})y/\sigma^2$ has a central chi-square distribution $\chisq (s-s^\star)$. If $h\alpha<1$, using moment generating function of central chi-square, the first expectation term in \eqref{eq:S.posterior} can be written as,
$$
\E_{\beta^\star}^{1/h}\bigl[\exp\{\tfrac{h \alpha}{2\sigma^2}y^\top (P_S-P_{S^\star})y\}\bigr]=(1-h\alpha)^{-(s-s^\star)/2h}.
$$
Second, from the spectral decomposition $X_{S^\star}^\top X_{S^\star}=\Gamma_{S^\star}\Lambda_{S^\star}\Gamma_{S^\star}^\top$, if $u=\Lambda_{S^\star}^{1/2}\Gamma_{S^\star}^\top \hat{\beta}_{S^\star}/\sigma$, then,
\begin{align*}
   \hat{\beta}_{S^\star}^\top Q_{S^\star}^{-1}\hat{\beta}_{S^\star}/\sigma^2
    =u^\top \big(I_s+\alpha g k_{S^\star} \Lambda_{S^\star}^{\lambda+1}\big)^{-1} u=\sum_{i=1}^{s^\star}\frac{1}{1+\alpha g k_{S^\star} d_{S^\star, i}^{\lambda+1}} u_i^2.
\end{align*}
Obviously, $u \sim \N(\Lambda_{S^\star}^{1/2}\Gamma_{S^\star}^\top \beta_{S^\star}^\star/\sigma, I_{s^\star})$, and this implies $u_i \overset{iid}{\sim} \N(d_{S^\star, i}^{1/2}\Gamma^\top_{S^\star, i}\beta_{S^\star}^\star/\sigma, 1)$.  It follows that the $u_i^2$s are independent non-central chi-square random variables, with non-centrality parameter $\mu_i=d_{S^\star, i}^{1/2}\Gamma^\top_{S^\star, i}\beta_{S^\star}^\star/\sigma$ and degrees of freedom of $1$. Using the same argument as in \eqref{eq:Zsolution}--\eqref{eq:Zbound}, we get,
\begin{align*}
\E_{\beta^\star}^{1/q}\bigl[\exp\{\tfrac{\alpha q(1-\phi)^2}{2\sigma^2}\hat{\beta}_{S^\star}^\top Q_{S^\star}^{-1}\hat{\beta}_{S^\star}\}\bigr] \le (1-t)^{-\frac{s^\star}{2q}}\exp\big\{\alpha n\|\beta_{S^\star}^\star\|^2 \lambda_{\max}(S^\star)\tfrac{t}{1-t}\big\},
\end{align*}
where $0 < t \lesssim (1-\phi)^2$.  
Since $t = O(n^{-1}) = o(1/s^\star)$, we have 
\[ (1-t)^{-s^\star}=\bigl(1+\tfrac{t}{1-t}\bigr)^{s^\star}=O(1).\]
Also, by Assumptions~\ref{asmp:2} and \ref{asmp:3}, the exponential term is asymptotically bounded by $e^{\alpha s^\star}$.  Therefore, we can conclude that the second expectation term in \eqref{eq:S.posterior} is also asymptotically upper-bounded by $e^{\alpha s^\star}$.  

Next, it is clear that 
\[ R_S \leq \prod_{i=1}^{s^\star} (1 + \alpha g k_{S^\star} d_{S^\star,i}^{\lambda+1})^{1/2}. \]
Since $k_S d_{S,i}^{\lambda+1}$ satisfies
\[\kappa(S)^{-(\lambda+1)} \leq k_S d_{S,i}^{\lambda + 1} \leq \kappa(S)^{\lambda + 1} , \]
and the lower and upper bounds are stable according to Assumption~\ref{asmp:3}, it follows that $R_S$ is upper-bounded by $m^{s^\star}$, where $m = (1+\alpha g \kappa(S^\star)^{\lambda+1})^{1/2}$.  


For $F_S$, as in Appendix~\ref{SS:proof1}, we have  
\[ \pi(S) \le {p \choose s}^{-1}\omega(s)^{\frac{\lambda}{2}}f_n(s) \quad \text{and} \quad \pi(S^\star) \ge {p \choose s^\star}^{-1}\omega(s^\star)^{-\frac{\lambda}{2}}f_n(s^\star).\]
Now we can bound $\E_{\beta^\star}\{\pi^n(S)\}$ above by, 
\[ e^{Gs^\star}\omega(s)^{\lambda}{p \choose s}^{-1}{p \choose s^\star}\big(\frac{z}{cp^a}\big)^{s-s^\star}\times O(1),\] 
where $G = \alpha + \log m$ and $z > 0$ is a constant.  By Theorem~\ref{thm:dimension}, we only need to consider configurations of size no more than $\rho s^\star$, where $\rho>\rho_0$ in \eqref{eq:rho0} and $a>r(1+\lambda)$.  Therefore,
\begin{align}
\E_{\beta^\star}\{\Pi^n(\beta:S_\beta \supset S_{\beta^\star})\} & =\sum_{S \supset S^\star}\E_{\beta^\star}\{\pi^n(S)\} \nonumber \\
& \leq e^{Gs^\star}\omega(\rho s^\star)^{\lambda}\sum_{s=s^\star+1}^{\rho s^\star}\frac{{p-s^\star \choose p-s}{p\choose s^\star}}{{p\choose s^\star}}\big(\frac{z}{cp^a}\big)^{s-s^\star} \label{eq:part1bound}
\end{align}
Under Assumption~\ref{asmp:3},  $\omega(\rho s^\star)$ is bounded, so following the results in \cite{c12}, \eqref{eq:part1bound} turns out to be,
\[ \E_{\beta^\star}\{\Pi^n(\beta:S_\beta \supset S_{\beta^\star})\} \le  \frac{e^{Gs^\star}s^\star}{p^{a}}\times O(1). \]
So, if $p$ and $a$ are such that $p^{a} \gg s^\star e^{Gs^\star}$, then $\E\{\Pi^n(S \supset S^\star)\}$ will go to $0$ as $n \to \infty$.

\subsubsection*{Part 2}
Consider a configuration $S$ satisfying $S \nsupseteq S^\star$, in Part 1, we already have \eqref{eq:S.posterior}, which is
 \begin{equation*}
        \E_{\beta^\star}\{\pi^n(S)\} \le F_SR_S \E_{\beta^\star}^{1/h}\bigl[\exp\{\tfrac{\alpha h}{2\sigma^2}y^\top (P_S-P_{S^\star})y\}\bigr]\E_{\beta^\star}^{1/q}\bigr[\exp\{\tfrac{\alpha q(1-\phi)^2}{2\sigma^2}\hat{\beta}_{S^\star}^\top Q_{S^\star}^{-1}\hat{\beta}_{S^\star}\}\bigr].
\end{equation*}
For the first expectation term, if we plug in $y=X_{S^\star}\beta_{S^\star}+\sigma z$, where $z \sim \N(0, I_n)$, then according to the results in \citet{c12}
\begin{align*}
y^\top(P_S-P_{S^\star})y & =-\|(I-P_S)X_{S^\star}\beta_{S^\star}^\star\|^2-2\sigma z^\top(I-P_S)X_{S^\star}\beta_{S^\star}^\star+\sigma^2 z^\top (P_S-P_{S^\star})z \\
& \le -\|(I-P_S)X_{S^\star}\beta_{S^\star}^\star\|^2-2\sigma z^\top (I-P_S)X_{S^\star}\beta_{S^\star}^\star+\sigma^2 z^\top (P_S-P_{S\cap S^\star})z. 
\end{align*}
Since $z^\top(I-P_S)X_{S^\star}\beta_{S^\star}^\star$ and  $z^\top (P_S-P_{S\cap S^\star})z$ are independent, using the moment generating function of normal and chi-square distributions, we can get,
\begin{align}
\E_{\beta^\star}^{1/h}\bigl[\exp\{&\tfrac{\alpha h}{2\sigma^2}y^\top (P_S-P_{S^\star})y\} \bigr] \notag \\
& \le (1-h\alpha)^{-\frac{1}{2h}(|S|-|S \cap S^\star|)}\exp\{-\tfrac{\alpha(1-h\alpha)}{2\sigma^2}\|(I-P_S)X_{S^\star}\beta_{S^\star}^\star\|^2\} \label{eq:noncentral}
\end{align}
With some algebraic manipulation, we can show that \[\|(I-P_S)X_{S^\star}\beta_{S^\star}^\star\|^2=\|(I-P_S)X_{S^\star\cap S^c}\beta_{S^\star \cap S^c}^\star\|^2,\]
and then based on Lemma 5 of \cite{arias2014}, we get
\[
\|(I-P_S)X_{S^\star\cap S^c}\beta_{S^\star \cap S^c}^\star\|^2 \ge n\ell(s^\star)\|\beta_{S^\star \cap S^c}^\star\|^2.
\]
By the beta-min condition \eqref{eq:beta.min}, it follows that $\|\beta_{S^\star \cap S^c}^\star\|^2 \ge \varrho_n^2(s^\star-|S^\star\cap S|)$, and, hence, 
\[ \exp\bigl\{-\tfrac{\alpha(1-h\alpha)}{2\sigma^2}\|(I-P_S)X_{S^\star}\beta_{S^\star}^\star\|^2 \bigr\} \le p^{-M(s^\star-|S^\star\cap S|)}. \]
Then following the results in \cite{c12},  we can get,
\[ \E_{\beta^\star}\{\Pi^n(\beta:S_\beta \nsupseteq S_{\beta^\star})\} \lesssim e^{Gs^\star} \sum_{s=0}^{\rho s^\star} \sum_{t=1}^{\min(s, s^\star)} \frac{{s^\star \choose t}{p-s^\star \choose s-t}{p \choose s^\star}}{{p \choose s}}(cp^a)^{s^\star-s}\{(1-h\alpha)p^{-M}\}^{s^\star-t}. \]
When $s < s^\star$, we have,
\[ \E_{\beta^\star}\{\Pi^n(\beta:S_\beta \nsupseteq S_{\beta^\star})\} \lesssim e^{Gs^\star} \sum_{s=0}^{s^\star-1}(cp^{1+a-M})^{s^\star-s} \lesssim \frac{m^{s^\star}}{p^{M-(a+1)}}, \quad M>a+1. \]
When $s \ge s^\star$, we have,
\[ \E_{\beta^\star}\{\Pi^n(\beta:S_\beta \nsupseteq S_{\beta^\star})\} \lesssim e^{Gs^\star}((1-h\alpha)p^{1-M}) \sum_{s=s^\star}^{\rho s^\star}(cp^{a-1})^{s^\star-s} \lesssim \frac{e^{Gs^\star}}{p^{M-(a+1)}}, \quad a>1. \]
Furthermore, if $p^{M-(a+1)}$ satisfies $p^{M-(a+1)} \gg e^{Gs^\star}$, the two fractions above will finally go to $0$ as $n$ and $p$ go to infinity.


\bibliographystyle{apalike} 
\bibliography{ecap9}

\end{document}